\newif\ifcamera
\cameratrue
\camerafalse %

\ifcamera
\documentclass[twoside,leqno]{article}
\else
\documentclass[11pt]{article}
\fi

\newif\ifsubmission
\submissiontrue
\submissionfalse %

\newif\ifcomment
\commenttrue
\commentfalse %

\ifsubmission
\commentfalse
\author{Anonymous Authors}
\else
\author{
  Yotam Kenneth-Mordoch
  \qquad
  Robert Krauthgamer%
  \thanks{The Harry Weinrebe Professorial Chair of Computer Science.
    Work partially supported by the Israel Science Foundation grant 1336/23,
    by the Israeli Council for Higher Education (CHE) via the Weizmann Data Science Research Center,
    and by a research grant from the Estate of Harry Schutzman.}
  \\[0.3em]
  Weizmann Institute of Science
  \\[0.3em]
  \texttt{\{yotam.kenneth,robert.krauthgamer\}@weizmann.ac.il}
}

\fi

\ifcamera
\usepackage[letterpaper]{geometry}
\usepackage{siamproceedings}
\else
\usepackage[letterpaper, margin=1in]{geometry}
\usepackage{amsthm}
\usepackage{amsmath}

\fi
\usepackage{amssymb}
\usepackage{algorithm}
\usepackage{algorithmicx}
\usepackage[noend]{algpseudocode}

\usepackage{tikz}
\usetikzlibrary{arrows.meta}

\usepackage{amsfonts}
\usepackage{import} 
\usepackage{xifthen}
\usepackage{mathtools}
\usepackage{arydshln}%
\usepackage{dsfont}
\usepackage[cal=esstix]{mathalfa}
\usepackage{thmtools}
\usepackage{thm-restate}
\usepackage{booktabs}
\usepackage{xspace}

\usepackage{comment}
\usepackage[colorinlistoftodos,textsize=tiny,textwidth=2cm,color=red!25!white,obeyFinal]{todonotes}

\ifcamera\else %
\usepackage[hypertexnames=false]{hyperref}
\hypersetup{colorlinks=true,allcolors=blue}
\usepackage{cleveref}
\fi

\AddToHook{cmd/appendix/before}{\crefalias{section}{appendix}}

\newcommand{\mintcut}{\mathrm{cut}}

\newcommand{\C}{\mathcal{C}}

\newcommand{\R}{\mathbb{R}}

\newcommand{\Exp}[2]{\mathbb{E}_{ #2}\left[ #1 \right]}

\newcommand{\Probability}[1]{\Pr\left[ #1 \right]}

\newcommand{\T}{\mathcal{T}}

\newcommand{\tO}{\tilde{O}}

\DeclareMathOperator{\poly}{poly}
\DeclareMathOperator{\polylog}{polylog}
\newcommand{\indic}[1]{ \mathds{1}_{\{#1\}} } %
\newcommand{\LCA}{\mathrm{LCA}}

\providecommand{\set}[1]{{\{#1\}}}
\newcommand{\eqdef}{\coloneqq}

\ifcomment

\else
\fi

\ifcamera
\newsiamthm{question}{Question}
\newsiamthm{claim}{Claim}
\newsiamthm{conjecture}{Conjecture}
\newsiamthm{observation}{Observation}
\newsiamthm{fact}{Fact}
\newsiamremark{remark}{Remark}

\else
\newtheorem{theorem}{Theorem}[section]

\newtheorem{claim}[theorem]{Claim}
\newtheorem{proposition}[theorem]{Proposition}
\newtheorem{lemma}[theorem]{Lemma}

\newtheorem{definition}[theorem]{Definition}
\newtheorem{observation}[theorem]{Observation}

\newtheorem{remark}[theorem]{Remark}
\fi

\newcommand{\insertG}{\mathtt{insert}_G}
\newcommand{\insertF}{\mathtt{insert}_F}
\newcommand{\deleteG}{\mathtt{delete}_G}
\newcommand{\deleteF}{\mathtt{delete}_F}
\newcommand{\query}{\mathtt{query}}

\title{Fully Dynamic Edge Connectivity in $\tilde{O}(n^{12/13})$ Time}
\date{}

\begin{document}

\maketitle

\ifcamera
\fancyfoot[R]{\scriptsize{Copyright \textcopyright\ 2027 by SIAM\\
Unauthorized reproduction of this article is prohibited}}
\else
\setcounter{page}{0}
\thispagestyle{empty}
\fi

\begin{abstract}
In the (fully) dynamic edge connectivity problem, the goal is to maintain the edge connectivity $\lambda_G$ of  an $n$-vertex graph $G$ that undergoes edge insertions and deletions.
Our main result is a randomized algorithm for maintaining edge connectivity in dynamic simple graphs
using worst-case update and query time $\tilde{O}(n^{12/13})$, for all values of $\lambda_G$.
This is the first algorithm that has $o(n)$ update and query time,
as all existing algorithms achieve this only when $\lambda_G$ is
below $n^{1/11}$ or above $n^{1/2}$ (up to polylogarithmic factors).
We then use the tools developed for this purpose to design two additional algorithms.
The first one is a deterministic algorithm for the exact same task,
that uses $n^{1+o(1)}$ worst-case update and query time
or $\tilde{O}(n)$ amortized update and query time;
this gives a polynomial improvement over existing deterministic algorithms.
The second one is a deterministic algorithm for the same task
but in dynamic unweighted multigraphs,
that uses $\tilde{O}(n^{3/2})$ worst-case update and query time.
\end{abstract}

\ifcamera\else
\newpage
\setcounter{page}{1}
\fi

\section{Introduction}
\label{sec:intro}
The \emph{edge connectivity} of an (unweighted) graph $G$, denoted $\lambda_G$, is the minimum number of edges that need to be removed from $G$ to disconnect it.
This fundamental graph parameter has been studied across a wide range of computational models, including the classical sequential setting~\cite{NI92,KS96,SW97,Kar00,GMW20,MN20}, parallel computation~\cite{GG18}, the distributed \textsc{Congest} model~\cite{GK13,DHN+19,DEM+21}, graph streams~\cite{MN20,AD21}, and the cut-query and quantum query models~\cite{RSW18,AEGLMN22,AL21}.
We focus on maintaining the edge connectivity in the fully dynamic setting, where the input graph $G=(V,E)$ on a fixed set of $n$ vertices that undergoes edge insertions and deletions.
In this setting, the goal is to minimize the time it takes to update the data structure after an edge insertion or deletion (called \emph{update time}) and the time it takes to query the edge connectivity (called \emph{query time}).
Throughout, the \emph{worst-case update time} is the maximum time taken by the algorithm to process any single update in any update sequence.
When we mention only the update time it is implied that the query time is the same.
Similarly, when the type of update time is not specified, we refer to the worst-case and not, say, amortized.

Maintaining edge connectivity in the fully dynamic setting has a rich history spanning over 40 years~\cite{Fre85,GI91}.
Early work focused on the regime where the edge connectivity $\lambda_G$ is a small constant, giving efficient algorithms for $\lambda_G \in \{1,2,3,4\}$~\cite{Fre85,GI91,EGIN97,HdLT01,KKM13}.
Recent work on the bounded connectivity regime has given efficient algorithms for values of $\lambda_G$ up to $ n^{o(1)}$ using $n^{o(1)}$ worst-case update time~\cite{JST24,ElHL25,ElHL26}.

The first algorithm for polynomial values of $\lambda_G$ was given in~\cite{Tho07};
it uses $\tO(\sqrt{n}\lambda_{\max}^7)$ worst-case update time, where $\lambda_{\max}$ is the maximum edge connectivity throughout the execution and $\tO(\cdot)$ hides polylogarithmic factors in $n$.
A similar technique was later used in \cite{dVC25} to obtain $\tO(\sqrt{n}\lambda_{\max}^{5.5})$ worst-case update time, which is sublinear in the number of vertices when $\lambda_{\max} < n^{1/11}/\polylog(n)$.
Unfortunately, when $\lambda_{\max}\ge n^{3/11}$, this bound is no better than $\tO(n^2)$, the trivial solution of recomputing the edge connectivity from scratch after each update.

Recently, several algorithms have achieved $\tO(n)$ worst-case update and query time for general $\lambda_G$, improving upon the trivial solution of $\tO(n^2)$ time~\cite{GHN+23,KK26a,HKMR25}.%
\footnote{The algorithm of \cite{HKMR25} achieves $\polylog(n)$ worst-case update time and $\tO(n)$ query time.}
In addition, the algorithm of \cite{KK26a} achieves $o(n)$ worst-case update time for $\lambda_G\ge n^{1/2}\polylog(n)$.
Hence, the only regime where $o(n)$ worst-case update time was not known is $\lambda_G$ between $n^{1/11}$ and $n^{1/2}$ (omitting polylogarithmic factors).
\emph{This left a tantalizing open question: can one achieve $o(n)$ worst-case update and query time for all values of $\lambda_G$?}
We answer this question affirmatively by presenting the first such algorithm for simple graphs.
In addition, our framework gives the first deterministic algorithm for maintaining edge connectivity in simple graphs with worst-case update time $n^{1+o(1)}$, and the first algorithm for multigraphs with worst-case update time $o(n^2)$ (both for general $\lambda_G$).

\subsection{Results}
Our main result is an algorithm that maintains the edge connectivity of an unweighted graph in the fully-dynamic setting with $\tO(n^{12/13})$ worst-case update and query time.
This is the first algorithm for this problem with worst-case update and query time $o(n)$ for all values of $\lambda_G$.
Previous algorithms achieve update and query time $o(n)$ only when $\lambda_G$ is small ~\cite{Tho07,dVC25,ElHL25,ElHL26} or for $\lambda_G  \ge n^{1/2}\polylog(n)$ \cite{KK26a}.
Existing algorithms for general $\lambda_G$ use $\tO(n)$ worst-case query time, although some achieve $\polylog (n)$ worst-case update time~\cite{GHN+23,KK26a,HKMR25}.
\begin{theorem}
    \label{theorem:sublinear-edge-connectivity}
    There exists an algorithm that maintains the edge connectivity $\lambda_G$ of a dynamic simple graph $G$ on $n$ vertices (namely undergoing edge insertions and deletions), using $\tO(n^{12/13})$ worst-case update and query time.
    The algorithm is randomized and succeeds with high probability against an oblivious adversary.
\end{theorem}
We note that this algorithm can also report the edges of the minimum cut by increasing the query time to $\tO(n^{12/13}+\lambda_G)$.
In addition, as long as the algorithm reports only the edge connectivity it is clearly pseudo-deterministic, and hence adversarially robust.

We also provide a \emph{deterministic} algorithm for maintaining the edge connectivity of a dynamic simple graph with $n^{1+o(1)}$ worst-case update time.
This algorithm too can also report the edges of the minimum cut using $O(n^{1+o(1)}+\lambda_G)$ time.
We note that this algorithm matches up to subpolynomial factors the best randomized algorithms known before this work \cite{GHN+23,KK26a,HKMR25};
the algorithm of \cite{HKMR25} achieves $\polylog(n)$ worst-case update time but $\tO(n)$ worst-case query time.
Previous deterministic algorithms use $n^{3/2+o(1)}$ worst-case update time or $\tO(\min \left\{ m^{11/12},m^{11/13}n^{1/13},n^{3/2} \right\})$ amortized update time, where $m$ is the number of edges in the graph \cite{Tho07,GHN+23,dVC25,KK26b}.
We note that faster deterministic algorithms exist for small values of $\lambda_G$~\cite{ElHL26}.
Finally, there exists a pseudo-deterministic algorithm that maintains a global minimum cut (i.e. it reports the edges of a minimum cut) of a dynamic simple graph with $\polylog (n)$ worst-case update time and $\tO(n)$ worst-case query time~\cite{Kenneth-Mordoch26}.
\begin{theorem}\label{theorem:deterministic-edge-connectivity}
    There exists a deterministic algorithm that maintains the edge connectivity $\lambda_G$ of a dynamic simple graph $G$ on $n$ vertices (undergoing edge insertions and deletions), using $n^{1+o(1)}$ worst-case update time or $\tO(n)$ amortized update time.
\end{theorem}

Using our framework, we also give the first algorithm that maintains the edge connectivity of a fully dynamic unweighted multigraph using $o(n^2)$ worst-case update time, without any restriction on $\lambda_G$.
\begin{theorem}\label{theorem:multigraph-edge-connectivity}
There exists a deterministic algorithm that maintains the edge connectivity $\lambda_G$ of a dynamic unweighted multigraph $G$ on $n$ vertices (namely undergoing edge insertions and deletions), using $\tO(n^{3/2})$ worst-case update and query time.
\end{theorem}

\subsection{Related Work}\label{sec:related}

\paragraph{Dynamic Approximate Minimum Cut.}
In contrast to the exact setting, where the $O(n^2)$ barrier for general $\lambda_G$ was broken only recently, sublinear algorithms for \emph{approximating} the edge connectivity have been known for some time.
An algorithm that maintains a multiplicative $\sqrt{2+o(1)}$-approximation of the edge connectivity using polylogarithmic amortized update time was provided in~\cite{TK00}.
The approximation ratio was later improved to $1+o(1)$, at the cost of an $\tO(\sqrt{n})$ worst-case update time~\cite{Tho07}.
In recent years, an algorithm nearly matching the best properties of both algorithms, i.e. giving an approximation ratio  $1+o(1)$ using $n^{o(1)}$ update time, was given in \cite{ElHL25,ElHL26}.

\paragraph{Non-trivial minimum cut sparsifiers.}
One of the main tools introduced in recent works on finding a minimum cut in a simple graph is the notion of a \emph{non-trivial minimum cut sparsifier}, a sparse graph that preserves all minimum cuts that are non-trivial, i.e. contain more than one vertex.
This idea was first introduced in deterministic sequential algorithms \cite{KT19,HRW17,Sar21}, and was later leveraged by randomized algorithms in several models of computation~\cite{GNT20,AD21,AEGLMN22}.
Constructions of such sparsifiers are central to all recent algorithms using $o(n^2)$ worst-case update time for maintaining the edge connectivity for general $\lambda_G$ in simple graphs~\cite{GHN+23,KK26a,HKMR25}.

\section{Technical Overview}
\label{sec:tech-overview}
We now give a technical overview of our algorithms and their analysis.
Our main technical contribution is a new framework for maintaining the edge connectivity of a dynamic graph based on dynamizing the minimum cut algorithm of \cite{MN20}.
Our framework achieves $o(n^2)$ worst-case update time for all values of $\lambda_G$.
In contrast, all previous sublinear algorithms (in $m$) for general $\lambda_G$ are based on constructing a non-trivial minimum cut sparsifier.
This technique is sufficient for obtaining our deterministic and multi-graph algorithms (\Cref{theorem:deterministic-edge-connectivity,theorem:multigraph-edge-connectivity}), but not for obtaining our main result of worst-case update time $o(n)$ for all values of $\lambda_G$ (\Cref{theorem:sublinear-edge-connectivity}).
Our second contribution is combining our framework with a non-trivial minimum cut sparsifier to yield our main result.

As mentioned above, our main technical result is a dynamic (and deterministic) version of the minimum cut algorithm of \cite{MN20}.
We dynamize that algorithm by incremental rebuilding, where we divide updates into epochs comprising $O(\lambda_G)$ updates each.
In each epoch we build a specialized data structure that is easy to update during the epoch, and allows answering edge connectivity queries without accessing the entire graph.
Towards the end of the epoch, we reconstruct the data structure in the background and then switch to the new data structure at the beginning of the next epoch.
Throughout, the \emph{age} of a data structure is the number of updates that have been applied to the graph since it was constructed, and a data structure is called \emph{valid} if its age is at most the number of updates it can handle.
We denote the  graph after $t$ updates by $G_t$ and the number of edges in $G_t$ by $m_t$.
The following theorem summarizes the guarantees of our algorithm.
\begin{theorem}[Dynamic Version of \cite{MN20}]
    \label{theorem:dynamic-mn20}
    There exists a deterministic algorithm that, given an initial weighted graph $G=G_0$ on $n$ vertices, constructs in time $\tO(m_0)$ a data structure that supports the following operations:
    \begin{enumerate}
        \item \texttt{update(e,w)}: Update the weight of an edge $e$ by $w\in \set{\pm 1}$ in $\polylog (n)$ time.
        \item \texttt{query()}: Return the edge connectivity of the current graph $G_t$ in $\tO(n)$ time.
    \end{enumerate}
    The data structure is valid for $c\cdot \lambda_{G_0}$ updates for some small constant $c>0$.
\end{theorem}

We now give an overview of the algorithm and its analysis.
The algorithm of \cite{MN20} is based on the $2$-respecting min-cut framework of \cite{Kar00}.
Given a tree $T=(V,E_T)$, a cut $C\subseteq V$ is said to $2$-respect $T$ if $|E_T(C,V\setminus C)|\le 2$, where $E_H(S,T)$ is the set of edges in a graph $H$ between vertex subsets $S,T\subseteq V$.
The key idea of this celebrated framework is that one can find in near-linear time a set of spanning trees of a graph $G$ such that every minimum cut $C$ of $G$ $2$-respects some tree $T$ in the packing.
We extend the original analysis of \cite{Kar00} to show that every cut $C$ that is near-minimum, namely $\mintcut_G(C)\le (1+\eta)\lambda_G$ for a small constant $\eta>0$ where $\mintcut_G(C)$ is the weight of the cut $C$, also $2$-respects some tree $T$ in the packing.
A similar result was previously shown in \cite{CQX19} in the context of bounding the number of approximate minimum cuts in the graph.
The following lemma proves this fact, and uses it to show that every minimum cut of $G_t$ $2$-respects some tree in a packing of $G_0$.
To make the algorithm deterministic, we leverage the deterministic skeleton graph construction of \cite{HLRW24}.
\begin{lemma}
    \label{lemma:tree-packing-stability}
    There exists a deterministic algorithm that, given a weighted graph $G=G_0$ constructs a weighted tree packing $\T$ of $G$ with $\polylog (n)$ trees in $\tO(m_0)$ time.
    Furthermore, for every $t\le c\cdot \lambda_{G_0}$, for a small enough constant $c>0$, if $C\subseteq V$ is a minimum cut of $G_t$ then $C$ $2$-respects some tree in the packing.
\end{lemma}

The second part of the framework is an efficient algorithm that given a tree $T$, finds a minimum cut that $2$-respects it.
This is where \cite{MN20} differs from the original algorithm of \cite{Kar00}, which relies on a clever dynamic programming technique to find the minimum cut among those that $2$-respects a given tree in near-linear $\tO(m)$ time.
The key insight of \cite{MN20} is that it suffices to check the value of $\tO(n)$ cuts instead of all $O(n^2)$ possible $2$-respecting cuts.
To efficiently query the values of these cuts, the algorithm constructs a data structure that allows querying the value of certain cuts in the graph, as follows.
\begin{definition}[Tree Cut-Query Oracle]
    \label{definition:T-cut-query-oracle}
    A \emph{tree cut-query oracle} is a data structure constructed for a rooted spanning tree $T$ of a weighted graph $G$, that supports querying the weight of the following cuts:
    \begin{enumerate}
        \item any cut that $2$-respects $T$; and
        \item any cut $C\subseteq V$ such that $C=\cup_{i\in [k]} S_i$ where each $S_i$ is a contiguous set of vertices encountered during a post-order traversal of $T$ and $k=O(1)$.
    \end{enumerate}
\end{definition}
\begin{lemma}
    \label{lemma:dynamic-T-cut-query-oracle}
    There exists a deterministic data structure for tree cut-query oracle that given an initial weighted graph $G=G_0$ and a fixed rooted spanning tree $T$ of $G_0$, with construction time $\tO(m_0)$.
    Furthermore, the data structure supports edge insertions/deletion to $G$ using $\polylog (n)$ worst-case update time, and can return the value of a cut in $\polylog (n)$ query time.
\end{lemma}
The original algorithm of \cite{MN20} builds a static tree cut-query oracle, based on the 2D semigroup range searching data structure of \cite{Chazelle88}.
It then adds sampling queries to the oracle, which allows it to find the $\tO(n)$ cuts relevant to query.
We obtain our dynamic algorithm, by using the dynamic version of the data structure of \cite{Chazelle88}.
We note that while it is possible to show that the random sampling method is robust to edge updates, and thus can be used in our dynamic algorithm, we require our framework to be deterministic; 
thus, we introduce a new deterministic procedure for finding the $\tO(n)$ cuts relevant to query.
Our procedure leverages the structural insights of \cite{MN20} to design a deterministic binary search procedure that finds the relevant cuts in $\tO(n)$ time using $O(\log n)$ cut queries per relevant cut.
Our deterministic procedure works also in the static setting, which may be of independent interest.
We encapsulate our deterministic variant of \cite{MN20} in the following lemma, whose proof appears in \Cref{sec:dynamic-mn20} and which immediately implies \Cref{theorem:dynamic-mn20}.
\begin{lemma}
    \label{lemma:encapsulation-of-mn20}
    There exists a deterministic algorithm that, given a spanning tree $T$ of a weighted graph $G$ and a corresponding tree cut-query oracle, finds in time $\tO(n)$ a minimum cut among those that $2$-respect $T$.
\end{lemma}
\begin{proof}[Proof of \Cref{theorem:dynamic-mn20}]
    Begin by constructing a tree packing $\T$ with $|\T|\le \polylog(n)$ of $G$ using \Cref{lemma:tree-packing-stability}.
    Then, for each tree $T\in \T$, construct a weighted tree cut-query oracle using \Cref{lemma:dynamic-T-cut-query-oracle}.
    When encountering an edge update, update each tree cut-query oracle accordingly.
    Finally, to answer a query for the edge connectivity of $G_t$, find the minimum cut that $2$-respects each tree $T\in\T$ using \Cref{lemma:encapsulation-of-mn20} and return the minimum among them.
    The crucial observation is that as long as $t\le c\cdot \lambda_G$, for an appropriate constant $c$, every minimum cut of $G_t$ $2$-respects some tree $T^*\in \T$ by \Cref{lemma:tree-packing-stability}
    Thus, the algorithm finds it, or another cut of equal or smaller value, when finding the minimum cut that $2$-respects $T^*$.

    To conclude the proof we analyze the time complexity of the algorithm.
    Updating each tree cut-query oracles takes $\polylog (n)$ time by \Cref{lemma:dynamic-T-cut-query-oracle}, and the total update time over all $T\in \T$ is $\polylog (n)$.
    In addition, the construction of the tree packing and all the tree cut-query oracles takes $\tO(m)$ time.
    We conclude that the overall time complexity of the algorithm is $\tO(m)$ for the construction and $\polylog (n)$ for each update.
\end{proof}

\paragraph{Additional Tools.}
Before proceeding to give an overview of the proofs of \Cref{theorem:sublinear-edge-connectivity,theorem:deterministic-edge-connectivity,theorem:multigraph-edge-connectivity}, we introduce several additional tools that are used throughout the paper.
The first is a maximal $k$-packing of forests.
\begin{definition}[Maximal $k$-Packing of Forests]
    \label{def:packing-of-forests}
    Given a graph $G = (V, E)$, a \emph{maximal $k$-packing of forests} is a collection of edge-disjoint forests $T_1, \ldots, T_k$ of $G$ such that for every $i \in [k]$ and every edge $e \in E \setminus \bigcup_{j\le i} T_j$, the set $T_i \cup \{e\}$ contains a cycle.
\end{definition}
By the classical result of Nagamochi and Ibaraki \cite{NI92}, a maximal $k$-packing of forests yields a subgraph with at most $k(n-1)$ edges that preserves all cuts of value at most $k$ in $G$, and is called a $k$-NI sparsifier.
The following lemma summarizes several algorithms for maintaining a maximal packing of forests, which appear in previous work on dynamic graph algorithms, see e.g. \cite{ADK+16,KK26a}.
We provide a proof of this lemma for completeness in \Cref{sec:packing-of-forests}.
\begin{lemma}
    \label{lemma:packing-of-forests}
    There exist fully dynamic algorithms that maintain a maximal $k$-packing of forests of a dynamic unweighted (multi)graph $G$ on $n$ vertices with the following guarantees:
    \begin{itemize}
        \item A deterministic algorithm with worst-case update time  $k\cdot n^{o(1)}$.
        \item A deterministic algorithm with amortized update time  $k\cdot \polylog (n)$.
        \item A randomized algorithm with worst-case update time  $k\cdot \polylog (n)$.
    \end{itemize}
\end{lemma}
The following lemma shows that maximal $k$-packings of forests are robust to edge updates.
It allows us to work on a $k$-NI sparsifier instead of $G$ itself, and thus decreases the processing time of our algorithms by reducing the number of edges.
The proof of this lemma is provided in \Cref{sec:robustness-edge-updates}.
\begin{lemma}
    \label{lemma:might-as-well-ni}
    Let $G=(V,E)$ be a graph, let $H$ be the union of a maximal $c\cdot \lambda_G$-packing of forests of $G$ for some $c>1$, and let $I=\set{e_1,\ldots,e_k}$ and $D=\set{e_1,\ldots,e_l}$ be sets of edge insertions and deletions respectively.
    If $k+l \le \tfrac{c-1}{2}\lambda_G$, then every minimum cut of $(H\cup I)\setminus D$ is a minimum cut of $(G\cup I)\setminus D$, and vice versa.
\end{lemma}

\subsection{Sublinear Fully Dynamic Edge Connectivity in Simple Graphs}
In this section we give an overview of the proof of \Cref{theorem:sublinear-edge-connectivity}, the complete proof appears in \Cref{sec:sublinear-edge-connectivity}.
Our main technical tool is an algorithm that maintains the edge connectivity of a fully dynamic simple graph with $\tO(n/\delta_G)$ worst-case update time, where $\delta_G$ is the minimum degree of $G$ (at the beginning of every epoch).
To obtain a worst-case update time of $\tO(n^{12/13})$, we combine it with the following algorithm of \cite{dVC25}.
\begin{theorem}[Theorem 1 of \cite{dVC25}]
    \label{theorem:small-lambda-edge-connectivity}
    There exists a deterministic algorithm that, given a dynamic simple graph $G$ on $n$ vertices undergoing edge insertions and deletions and an upper bound $\lambda_{\max}$ on the edge connectivity, maintains the edge connectivity $\lambda_G$ with $\tO(\sqrt{n}\lambda_{\max}^{5.5})$ worst-case update time if $\lambda_G \leq \lambda_{\max}$ always.
\end{theorem}
We obtain \Cref{theorem:sublinear-edge-connectivity} by using both algorithms, balanced by choosing $\lambda_{\max} = n^{1/13}$.
To ensure that the edge connectivity of the graph on which we apply the algorithm of \Cref{theorem:small-lambda-edge-connectivity} never increases beyond $n^{1/13}$, we maintain a maximal $n^{1/13}$-packing of forests of the graph.
Then, running \Cref{theorem:small-lambda-edge-connectivity} on the packing ensures that $\lambda_{\max} \le n^{1/13}$, and thus the algorithm maintains the edge connectivity (up to $n^{1/13}$) using $\tO(n^{12/13})$ worst-case update time.

We now provide an overview of the proof of our $\tO(n/\delta_G)$ worst-case update time algorithm.
The algorithm is based on our dynamic version of the \cite{MN20} minimum cut algorithm.
Unfortunately, the query time of our dynamic algorithm is linear in $n$ and applying it directly cannot yield worst-case update time $o(n)$.
We bypass this issue by decreasing the number of vertices in the graph while preserving the edge connectivity.
To do this we use a non-trivial (near) minimum-cut sparsifier (NMC sparsifier), a widely used tool in many minimum cut algorithms for simple graphs \cite{KT19,HRW17,GNT20,AD21,GHN+23,HKMR25,KK26a}.
\begin{definition}[NMC Sparsifier]
    \label{definition:weak-nmc}
    A \emph{weak (strong) non-trivial minimum cut sparsifier} of a graph $G=(V,E)$ is a contraction $H=(V_H,E_H)$ of $G$ such that for every cut $C\subseteq V$ in $G$ with $|C|>1$ and $\mintcut_G(C) \le \lambda_G + \frac{1}{10}\delta_G$, where $\delta_G$ is the minimum degree of $G$, no edge in $E(C,V\setminus C)$ is contracted in $H$ with constant probability (with probability $1$).%
    \footnote{The value $1/10$ is arbitrary and can be replaced by any small enough constant.}
\end{definition}
Note that previous work in the dynamic setting only used the fact that NMC sparsifiers preserve minimum cuts, this is sufficient for running a static algorithm on the sparsifier.
However, to use our framework of \Cref{theorem:dynamic-mn20} we need to ensure that the minimum cut of the sparsifier is preserved for $\Theta(\lambda_G)$ updates, which is not guaranteed by the usual definition of NMC sparsifiers.
Therefore, our algorithm leverages the fact that the contraction preserves near-minimum cuts, namely cuts of value at most $\lambda_G + \frac{1}{10}\delta_G$, which allows us to prove the following analog of \Cref{lemma:might-as-well-ni} for NMC sparsifiers.
The proof is provided in \Cref{sec:robustness-edge-updates}.
\begin{lemma}
    \label{lemma:might-as-well-nmc}
    Let $G$ be a graph, let $H$ be a weak (strong) NMC sparsifier of $G$, and let $I=\set{e_1,\ldots,e_k}$ and $D=\set{e_1,\ldots,e_l}$ be sets of edge insertions and deletions respectively.
    If $k+l \le \frac{1}{20}\delta_G$, then for each minimum cut $C\subseteq V$ of $G'\coloneqq (G\cup I)\setminus D$ such that $|C|>1$ with constant probability (with probability $1$), no edge in $E_{G'}(C,V\setminus C)$ is contracted in $(H\cup I)\setminus D$.
\end{lemma}

Most NMC constructions have $\tO(n/\delta_G)$ vertices and $\tO(n)$ edges, which allows us to apply our dynamic minimum cut algorithm on the contracted graph with $\tO(n/\delta_G)$ query time.
However, the dynamic data structure of \Cref{theorem:dynamic-mn20} on an NMC sparsifier $H$ takes $\tO(|E_H|)=\tO(n)$ time to construct, and remains valid for only $\Theta(\lambda_G)$ updates, yielding worst-case update time $\tO(n/\lambda_G)$ which can be much worse than $\tO(n/\delta_G)$.
To overcome this barrier, we apply our framework on a maximal $2\lambda_G$-packing of forests of the NMC sparsifier.
By combining \Cref{lemma:might-as-well-ni,lemma:might-as-well-nmc}, this subgraph of the NMC sparsifier preserves the minimum cut, and thus we can use our data structure on this subgraph.
In addition, since the maximal $2\lambda_G$-packing of forests has $\tO(n \lambda_G/\delta_G)$ edges, the construction time of our data structure is $\tO(\lambda_G n/\delta_G)$, yielding worst-case update time $\tO(n/\delta_G)$ by incremental rebuilding over $\Theta(\lambda_G)$ updates.
Finally, since each NMC sparsifier only preserves each cut with constant probability, we build $O(\log n)$ NMC sparsifiers in parallel to ensure that with high probability at least one repetition succeeds.
We state the guarantees of our algorithm in the following lemma, whose proof appears in \Cref{sec:sublinear-edge-connectivity}.
\begin{lemma}
    \label{lemma:sublinear-edge-connectivity}
    There exists an algorithm that, given a dynamic simple graph $G$ on $n$ vertices undergoing edge insertions and deletions and a degree threshold $\Delta$, maintains the edge connectivity $\lambda_G$ whenever $\delta_G \ge \Delta$, with $\tO(\min\left\{ n/\delta_G, n/\Delta \right\})$ worst-case update and query time, where $\delta_G$ is the minimum degree of $G$.
    The algorithm is randomized and succeeds with high probability against an oblivious adversary.
\end{lemma}
\begin{proof}[Proof of \Cref{theorem:sublinear-edge-connectivity}]
    The algorithm maintains a maximal $n^{1/13}$-packing of forests of the graph using $\tO(n^{1/13})$ worst-case update time by \Cref{lemma:packing-of-forests}, and runs the algorithm of \Cref{theorem:small-lambda-edge-connectivity} on the packing.
    By the properties of a maximal packing of forests, the edge connectivity of the packing is at most $n^{1/13}$, and thus the algorithm of \Cref{theorem:small-lambda-edge-connectivity} maintains the edge connectivity, up to $n^{1/13}$, with $\tO(n^{12/13})$ worst-case update time.

    In parallel, run the algorithm of \Cref{lemma:sublinear-edge-connectivity} on the original graph with degree threshold $\Delta = n^{1/13}$, which maintains the edge connectivity of the original graph whenever it is at least $n^{1/13}$ (as $\delta_G \ge \lambda_G$) with $\tO(n^{12/13})$ worst-case update time.
    To answer a query of the edge connectivity of the graph, return the result of \Cref{theorem:small-lambda-edge-connectivity} if it is strictly smaller than $n^{1/13}$, and the edge connectivity maintained by the algorithm of \Cref{lemma:sublinear-edge-connectivity} otherwise.
\end{proof}

\subsection{Deterministic Dynamic Edge Connectivity in Simple Graphs}
\label{sec:deterministic-alg-overview}

Our deterministic algorithm is similar to our randomized algorithm for simple graphs, but instead of using an NMC sparsifier it operates directly on an NI sparsifier of the graph.
We provide an alternative self-contained proof of this result in \Cref{sec:simple-deterministic}.

\begin{proof}[Proof of \Cref{theorem:deterministic-edge-connectivity}]
    The algorithm maintains a maximal $n$-packing of forests of the graph.
    Throughout, assume that the edge connectivity is at least $\kappa=\polylog (n)$, and otherwise we use the following simple algorithm.
    Begin by maintaining a maximal $\kappa=\polylog (n)$-packing of forests of the graph, and then after each update, apply some static algorithm (e.g. \cite{GMW20,MN20}) to query its edge connectivity.
    This uses $\tO(n)$ worst-case update time by \Cref{lemma:packing-of-forests}, and does not affect the overall time complexity of the algorithm.
    We then only instantiate the data structure of \Cref{theorem:dynamic-mn20} only when the edge connectivity is at least $\kappa/2$.

    Divide the updates into epochs of $\lambda_G/40$ updates.
    Assume that at the beginning of each epoch the algorithm is given the data structure of \Cref{theorem:dynamic-mn20} of a $2\lambda_G$-NI sparsifier of $G$ whose age is less than $\lambda_G/80$.
    Throughout the epoch, whenever an edge is inserted or deleted in $G$, the algorithm inserts/deletes updates the data structure accordingly.
    By \Cref{lemma:might-as-well-ni}, we have that every non-trivial minimum cut of $G_t$ is a minimum cut of the $2\lambda_G$-NI sparsifier, for every $t$ in the epoch.
    Therefore, to answer a minimum cut query, it suffices to query the data structure, and return the minimum between it and the minimum degree of $G_t$.

    When only $\lambda_G/80$ updates are left in the epoch, use the first $2\lambda_G$ forests in the packing to build a $2\lambda_G$-NI sparsifier of $G$, and construct the data structure of \Cref{theorem:dynamic-mn20} on the $2\lambda_G$-NI sparsifier.
    Note that the algorithm can always know $\lambda_G$ since the update and query times of the algorithm are the same.
    Constructing the data structure takes $\tO(n\lambda_G)$ time by \Cref{theorem:dynamic-mn20}, and thus one can spend $\tO(n)$ time during $\lambda_G/160$ updates to build the data structure in the background.
    During the last $\lambda_G/160$ updates of the epoch, apply all the queued updates to the new data structure, and then switch to it at the beginning of the next epoch.
    Notice that the age of the data structure at the end of the epoch is at most $\lambda_G/80$, and thus it is valid for the next epoch.
    Finally, observe that the algorithm constructing the data structure requires a static copy of the maximal $2\lambda_G$-packing of forests, which we obtain by queueing the updates to the maximal $2\lambda_G$-packing of forests and applying them in the next $\lambda_G/100$ updates after the data structure is constructed.

    Finally, the algorithm uses $n^{1+o(1)}$ worst-case update time or $\tO(n)$ amortized update time, depending on the choice of the algorithm for maintaining the maximal $n$-packing of forests by \Cref{lemma:packing-of-forests}, which concludes the proof.
\end{proof}

\subsection{Dynamic Edge Connectivity in Multigraphs}
\label{sec:multigraph-overview}
In this section we give the proof of \Cref{theorem:multigraph-edge-connectivity}.
The proof is based on splitting into two cases by the value of the edge connectivity.
When the connectivity is small, it can be found by running a static algorithm on a maximal packing of forests.
The complementary case, when the connectivity is large, uses our dynamic version of the \cite{MN20} minimum cut algorithm directly on the graph.
\begin{proof}[Proof of \Cref{theorem:multigraph-edge-connectivity}]
    The algorithm maintains a maximal $2\sqrt{n}$-packing of forests of $G$, and a weighted version of $G$ that merges parallel edges into a single edge with weight equal to the number of parallel edges.
    Throughout, we assume that the algorithm knows the edge connectivity $\lambda_G$ of the graph before the update, this is achieved trivially since the update and query times of our algorithm are the same.

    The algorithm is split into epochs of length $c \sqrt{n}/2$ updates, where $c>0$ is the constant from \Cref{theorem:dynamic-mn20}.
    Assume that at the beginning of each epoch, if $\lambda_G >\sqrt{n}$, then the algorithm is given a data structure of \Cref{theorem:dynamic-mn20} on the weighted version of $G$ whose age is less than $c\sqrt{n}/10$ (this is obtained by constructing the data structure in the background during the previous epoch).
    During the epoch the algorithm updates the data structure according to the edge updates, and answers queries using that data structure.
    Otherwise, the algorithm answers queries by running an offline near-linear-time minimum cut algorithm on the $2\sqrt{n}$-packing of forests.
    Notice that the data structure is valid throughout the entire epoch since $c\sqrt{n} (1/10+1/2) < c\sqrt{n}$, and thus the algorithm returns the correct answer to all queries.

    When  only $c\sqrt{n}/10$ updates remain in the epoch, the algorithm checks if $\lambda_G > (1-c/2)\sqrt{n}$, and if this occurs it constructs the data structure of \Cref{theorem:dynamic-mn20} on the weighted version of $G$.
    Observe that this guarantees that if $\lambda_G > \sqrt{n}$ at the beginning of the next epoch, then the data structure exists and is valid for the next epoch.
    When constructing the data structure, the algorithm spends at most $\tO(n^{3/2})$ time after each update to build the data structure in the background (spread over $O(\sqrt{n})$ updates), and then switches to the new data structure at the beginning of the next epoch.
    Notice that the age of the data structure at the end of the epoch is at most $c\sqrt{n}/10$, and thus it is valid for the next epoch.
    Finally, the algorithm constructing the data structure requires a static copy of the weighted version of $G$, which we obtain by queueing the edge updates and applying them in the next $c\sqrt{n}/100$ updates after the data structure is constructed.

    To conclude we analyze the time complexity of the algorithm.
    The weighted version of $G$ can be easily maintained using $O(1)$ time per update, and the maximal $2\sqrt{n}$-packing of forests can be maintained using $\tO(n^{1/2})$ randomized or $n^{1/2+o(1)}$ deterministic worst-case update time by \Cref{lemma:packing-of-forests}.
    In addition, we spend at most $\tO(n^{3/2})$ time after each update to build the data structure of \Cref{theorem:dynamic-mn20} in the background, and thus the worst-case update time is $\tO(n^{3/2})$.
    The query time is $\tO(n^{3/2})$ when $\lambda_G \le \sqrt{n}$, and $\tO(n)$ when $\lambda_G > \sqrt{n}$, which is always at most $\tO(n^{3/2})$.
    Therefore, the algorithm uses $\tO(n^{3/2})$  worst-case update time and query time, which concludes the proof.
\end{proof}

\subsection{Future Work}

\paragraph{Faster Dynamic Edge Connectivity.}
While our work achieves $o(n)$ worst-case update time for all values of $\lambda_G$, we do not believe that $\tO(n^{12/13})$ is optimal.
Known algorithms using $o(n^{12/13})$ update time work only for restricted ranges of $\lambda_G$.
In particular, when $\lambda_G$ is very small (up to $n^{o(1)}$) or very large (at least $n^{1-o(1)}$) one can achieve $n^{o(1)}$ update time, see \cite{JST24,ElHL26,KK26a} and \Cref{lemma:sublinear-edge-connectivity}.
A natural direction is to extend the range of $\lambda_G$ for which $o(n^{12/13})$ update time is achievable.

\paragraph{Polylogarithmic Update Time.}
While it seems that the query time of our algorithm inherently relies on the framework of \Cref{theorem:dynamic-mn20}, which cannot improve upon $\tO(n/\delta_G)$, we believe it is possible to improve the update time.
The framework requires maintaining a weighted tree packing of $G$, and a tree cut-query oracle for each tree in the packing.
Maintaining a weighted tree packing of size $\polylog(n)$ using $\polylog(n)$ worst-case update time is achievable using existing tools.%
\footnote{For example, one can construct a skeleton graph \cite{Kar00} by subsampling $G$ and maintain a dynamic tree packing of the resulting graph using existing results \cite{Tho07,dVC25}.}
Therefore, we ask whether a tree cut-query oracle under tree updates (i.e. edge insertions/deletions in the tree) with $\polylog(n)$ worst-case update time exists.

\paragraph{Deterministic $o(n)$ Edge Connectivity.}
There is only one use of randomness in our $\tO(n/\delta_G)$ worst-case update time algorithm that is hard to derandomize: maintaining the clusters of an NMC sparsifier.
Unfortunately, existing deterministic NMC constructions are based on static algorithms using $\tO(m)$ time that are not easily amenable to the dynamic setting.
A deterministic dynamic algorithm maintaining the clusters of an NMC sparsifier using $o(n)$ worst-case update time would, through the framework of \Cref{theorem:dynamic-mn20}, yield a deterministic $o(n)$ worst-case update time algorithm for all values of $\lambda_G$.

\section{Preliminaries}\label{sec:prelim}
Throughout we use the following easy observation.
\begin{observation}
    \label{observation:cut-query-edges}
    It is possible to find the weight of the set of edges between any two sets $S,T\subseteq V$ by querying the cuts $\mintcut_G(S),\mintcut_G(T),\mintcut_G(S\cup T)$
\end{observation}
\begin{proof}
    The proof follows from $\mintcut_G(S\cup T) = \mintcut_G(S) + \mintcut_G(T) - 2\cdot w(E(S,T))$.
\end{proof}

\paragraph{Heavy-Light Decomposition.}
Throughout we use the heavy-light decomposition of \cite{ST83}.
Throughout, let $T$ be a tree rooted at an arbitrary vertex $r$, and for a vertex $v$ let $v^{\downarrow}$ denote the subtree of $T$ rooted at $v$ and $|v^{\downarrow}|$ the number of vertices in it.
For a non-leaf vertex $v$, its \emph{heavy child} is the child $u$ maximizing $|u^{\downarrow}|$ (breaking ties arbitrarily), the edge $(v,u)$ is the corresponding \emph{heavy edge}, and every other edge from $v$ to a child is a \emph{light edge}.
The paths formed by heavy edges are called \emph{heavy paths}, and the decomposition of $T$ into heavy paths is a \emph{heavy-light decomposition}.
Notice that a light edge $(v,u)$ satisfies $|u^{\downarrow}|\le |v^{\downarrow}|/2$, so every root-to-leaf path of $T$ contains at most $\log n$ light edges, and hence at most $\log n$ paths in the decomposition.
Finally, to simplify the analysis, we append to each light path the edge connecting its topmost vertex to its parent.
\begin{theorem}[\cite{ST83}]
    \label{fact:heavy-light}
    There exists an algorithm that, given a tree $T$ on $n$ vertices, computes a heavy-light decomposition of $T$ in $O(n)$ time.
\end{theorem}

\section{Dynamic Minimum Cut Algorithm}
\label{sec:dynamic-mn20}
In this section we prove the main components of our dynamic version of the algorithm of \cite{MN20} for finding a global minimum cut in a weighted graph.
The section is organized as follows.
In \Cref{sec:interested-pairs} we show how to find the potential cuts that the algorithm of \cite{MN20} needs to check.
Then, in \Cref{sec:tree-packing-stability} we prove \Cref{lemma:tree-packing-stability}, which shows that tree packing is stable for $O(\lambda_G)$ edge updates, and in \Cref{sec:dynamic-cut-query-oracle} we show how to maintain a tree cut-query oracle under edge updates proving \Cref{lemma:dynamic-T-cut-query-oracle}.
Finally, in \Cref{sec:putting-it-all-together}, we put all the components together to prove \Cref{lemma:encapsulation-of-mn20}.

\subsection{Finding Potential Cuts}
\label{sec:interested-pairs}
The main tool towards identifying the minimum cuts that $2$-respect the tree are \emph{interested pairs of vertices}, introduced by \cite{MN20}.
Note that \cite{MN20} uses the term interested edges, however we find it more intuitive to refer to the vertices;
furthermore, the two notations are equivalent by arbitrarily rooting $T$, and identifying each edge with its lower endpoint.
We say that two vertices $u,v\in T$ are \emph{independent}, denoted by $u\perp v$, if they are not in the same root-to-leaf path in $T$.

Informally, a vertex $u$ is \emph{cross interested} in a vertex $v$, where $u\perp v$, if most of the edges leaving the subtree $u^{\downarrow}$ have an endpoint in $v^{\downarrow}$, hence the cut $u^{\downarrow}\cup v^{\downarrow}$ is potentially smaller than the cut $u^{\downarrow}$.
This holds up to the edges between $v^{\downarrow}$ and $V\setminus (u^{\downarrow}\cup v^{\downarrow})$, which are counted in the cut $u^{\downarrow}\cup v^{\downarrow}$ but not in $u^{\downarrow}$.
Similarly, $u$ is down-interested in $v$, where $v\in u^{\downarrow}$, if most of the edges leaving $u^{\downarrow}$ have an endpoint in $v^{\downarrow}$.
Again, this implies that the cut $u^{\downarrow}\setminus v^{\downarrow}$ is potentially smaller than the cut $u^{\downarrow}$.
We now formally define the notion of interested pairs of vertices.
\begin{definition}[Cross-Interested Vertices]
    A vertex $u\in T$ is \emph{cross interested} in $v\in T$ if $u\perp v$, and $w(E(v^{\downarrow}, u^{\downarrow}))>\mintcut_G(u^{\downarrow})/2$.
\end{definition}
\begin{definition}[Down-Interested Vertices]
    A vertex $u\in T$, is \emph{down interested} in $v\in u^{\downarrow}$ if  $w(E(V\setminus u^{\downarrow}, v^{\downarrow}))>\mintcut_G(u^{\downarrow})/2$.
\end{definition}
\begin{observation}[Observation 3.9 of \cite{MN20}]
    \label{observation:interested-edges-structure}
    Given a vertex $u$, there cannot be two different vertices $v_1,v_2$ where $u\perp v_1,v_2$, $v_1\perp v_2$ and $u$ is cross interested in both.
    Similarly, there cannot be two vertices $v_1,v_2\in u^{\downarrow}$ such that $u$ is down interested in both.
\end{observation}
Notice that all the vertices that $u$ is interested in are contained within a root-to-leaf path in $T$.
Hence, it suffices to find the top and bottom vertices in this path to find all interesting vertices.
The main technical result of this section leverages this insight to find all interested pairs of vertices in $\tO(n)$ time using a tree cut-query oracle.
\begin{lemma}
    \label{lemma:find-all-interesting-pairs}
    There exists a deterministic algorithm that, given a spanning tree $T$ of a graph $G$ and a tree cut-query oracle, finds all pairs of interesting vertices in $T$ in time $\tO(n)$.
\end{lemma}
The proof of \Cref{lemma:find-all-interesting-pairs} follows directly from the following claims, and the that a heavy-light decomposition can be constructed in $\tO(n)$ time by \Cref{fact:heavy-light}.
\begin{claim}
    \label{claim:find-down-interested}
    There exists a deterministic algorithm that, given a spanning tree $T$ of a graph $G$, a tree cut-query oracle, a heavy-light decomposition of $T$, and a vertex $u\in T$, finds all vertices $v\in T$ such that $u$ is down interested in $v$ in time $\polylog(n)$.
\end{claim}
\begin{claim}
    \label{claim:find-cross-interested}
    There exists a deterministic algorithm that, given a spanning tree $T$ of a graph $G$, a tree cut-query oracle, a heavy-light decomposition of $T$, and a vertex $u\in T$, finds all vertices $v\in T$ such that $u$ is cross interested in $v$ in time $\polylog(n)$.
\end{claim}
\begin{proof}[Proof of \Cref{lemma:find-all-interesting-pairs}]
    The algorithm iterates over all vertices $u\in T$, and for each vertex $u$ it finds all vertices $v\in T$ such that $u$ is down interested in $v$ by \Cref{claim:find-down-interested}, and all vertices $v\in T$ such that $u$ is cross interested in $v$ by \Cref{claim:find-cross-interested}.
    By \Cref{fact:heavy-light}, a heavy-light decomposition of $T$ can be constructed in $\tO(n)$ time, and thus the overall time complexity of the algorithm is $\tO(n)$.
\end{proof}
We now proceed to prove \Cref{claim:find-down-interested} and \Cref{claim:find-cross-interested}.
\begin{proof}[Proof of \Cref{claim:find-down-interested}]
    The proof is based on interleaving two binary searches.
    Recall that all vertices $v$ in which $u$ is down-interested are contained in the subtree rooted at $u$.
    Therefore, starting from $u$, first find the child $x$ of $u$ such that $u$ is down interested in $x$.
    Then, take the path $P$ in the heavy-light decomposition that contains the edge $(u,x)$, and perform a binary search on $P$ to find the last vertex $\alpha\in P$ that $u$ is interested in.
    Finally, repeat the same process in the tree rooted at $\alpha$.
    
    Let $r$ be the root of a subtree in $u^{\downarrow}$ that we start the process from.
    Denote the children of $r$ by $r_1,\ldots,r_k$ where the ordering is according to a post-order traversal of $T$.
    To find the child $r^*$ that $u$ is interested in, perform a binary search.
    Notice that there is such an $r^*$ exists if and only if $w(E(V\setminus u^{\downarrow}, r_i^{\downarrow}))>\mintcut_G(u^{\downarrow})/2$ for some child $r^*$.
    Otherwise, $u$ is not interested in any vertices in the subtree rooted at $r$.
    
    For every sequence $I=\{l,l+1,\ldots,q\}$ such that $1\le l< q \le k$, let $r_I^{\downarrow} =\cup_{i\in I} r_i^{\downarrow}$.
    Repeatedly splitting the children of $r$ into two halves $I_1,I_2$, finding the half for which $w(E(V\setminus u^{\downarrow}, r_{I_j}^{\downarrow}))$ is larger and recursing on it, terminates by returning $r^*$ if it exists.
    This holds since the half containing $r^*$ must have $w(E(V\setminus u^{\downarrow}, r_{I_j}^{\downarrow}))>\mintcut_G(u^{\downarrow})/2$.
    Furthermore, this process terminates in $\log n$ iterations since the number of children of $r$ is at most $n$.

    We now show how to implement this process using a tree cut-query oracle.
    By \Cref{observation:cut-query-edges}, $w(E(V\setminus u^{\downarrow}, r_I^{\downarrow}))$ can be computed by querying the cuts $(V\setminus u^{\downarrow}) \cup r_I^{\downarrow},r_I^{\downarrow},u^{\downarrow}$.
    Notice that $r_I$ is a contiguous set of vertices in the post-order traversal of $T$, and therefore each of these cuts is a union of at most two contiguous sets in the post-order traversal of $T$.
    Hence, they can be queried in $\polylog(n)$ time using the tree cut-query oracle.
    Therefore, the total time complexity of this process is $\polylog(n)$.

    We now proceed to find the last vertex $\alpha\in P$ that $u$ is interested in.
    By \Cref{observation:interested-edges-structure}, if $u$ is interested in some $x\in P$, then $u$ is interested also in every vertex before $x$ in $P$.
    Furthermore, $w(E(V\setminus u^{\downarrow}, x^{\downarrow}))$ can be computed by querying the cuts $(V\setminus u^{\downarrow}) \cup x^{\downarrow},x^{\downarrow},u^{\downarrow}$.
    Again, all these cuts are unions of at most two contiguous sets in the post-order traversal of $T$, and therefore can be queried in $\polylog(n)$ time using the tree cut-query oracle.
    Therefore, it is straightforward to perform a binary search on $P$ using $\log n$ queries to find the last vertex $\alpha\in P$ in which $u$ is interested.

    To bound the overall time complexity, note that every leaf to root path in $T$ intersects $O(\log n)$ paths in the heavy-light decomposition, and that throughout the algorithm we perform at most two binary searches on each of these paths, which yields a total time complexity of $\polylog(n)$.
\end{proof}
\begin{proof}[Proof of \Cref{claim:find-cross-interested}]
    The proof is similar to the previous one, interleaving two binary searches.
    Throughout, let $r$ be the root of the current subtree in which we are searching for vertices in which $u$ is interested.
    In addition, denote by $z$ the lowest vertex in which $u$ is interested in the subtree rooted at $r$.
    Notice that all vertices $v$ in which $u$ is cross-interested are exactly the vertices in the path from $z$ to $\LCA(z,u)$ (not including $\LCA(z,u)$), and therefore it is sufficient to find the top and bottom vertices in this path.
    Therefore, our goal is to first find $\LCA(z,u)$, and then to find $z$.

    We split into two cases.
    The first, is when $r$ is a strict ancestor of $\LCA(z,u)$, and the second is the complementary case.
    To find whether $\LCA(z,u)=r$, let $x$ be the child of $r$ that is an ancestor of $u$.
    Then, calculate $w(E(u^{\downarrow}, x^{\downarrow}\setminus u^{\downarrow}))$.
    Observe that if $w(E(u^{\downarrow}, x^{\downarrow}\setminus u^{\downarrow}))>\mintcut_G(u^{\downarrow})/2$, then $u$ is interested in some descendant of $x$, and hence $r$ is a strict ancestor of $\LCA(z,u)$.
    Otherwise, $\LCA(z,u)=r$.
    Note that to find $w(E(u^{\downarrow}, x^{\downarrow}\setminus u^{\downarrow}))$ it is sufficient to query the cuts $u^{\downarrow}\cup (x^{\downarrow}\setminus u^{\downarrow}),u^{\downarrow},x^{\downarrow}\setminus u^{\downarrow}$ by \Cref{observation:cut-query-edges}.
    As these cuts are unions of at most two contiguous sets in the post-order traversal of $T$, they can be queried in $\polylog(n)$ time using the tree cut-query oracle.

    In the case where $r$ is a strict ancestor of $\LCA(z,u)$, let $P=(r=\pi_1,\ldots,\pi_l)$ be the path containing the edge $(r,x)$ in the heavy-light decomposition.
    Then, let $P'$ be the path truncated at $\LCA(u,\pi_l)$, i.e. the last vertex in $P$ that is an ancestor of $u$.
    Next, perform a binary search to find the last vertex $\alpha\in P$ which is a (perhaps non-strict) ancestor of $\LCA(z,u)$.
    Similarly to the previous paragraph, $u$ is interested in a descendant of a vertex $y\in P'$ if and only if $w(E(u^{\downarrow}, y^{\downarrow}\setminus u^{\downarrow}))>\mintcut_G(u^{\downarrow})/2$, which can be checked by querying the cuts $u^{\downarrow}\cup (y^{\downarrow}\setminus u^{\downarrow}),u^{\downarrow},y^{\downarrow}\setminus u^{\downarrow}$.
    Again, these cuts can be queried in $\polylog(n)$ time using the tree cut-query oracle.
    Finally, recurse on the tree rooted at $\alpha$ until $\LCA(z,u)$ is found.

    Next we need to find the lowest vertex $z$ in which $u$ is interested.
    Notice that since $r$ is not a strict ancestor of $\LCA(z,u)$, $z$ must be a descendant of a child of $r$ that is not an ancestor of $u$.
    To find the child $x$ of $r$ such that $u$ is interested in $x$ (if one exists), perform a binary search.
    Denote the children of $r$ by $r_1,\ldots,r_k$ where the children are according to a post-order traversal ordering of $V$.
    In addition, let $r_j$ be the child of $r$ that is an ancestor of $u$, if it exists.
    For every sequence $I=\{l,l+1,\ldots,q\}$ such that $l,q\in [k]$, let $r_{I\setminus\{j\}}^{\downarrow} =\cup_{i\in I\setminus\{j\}} r_i^{\downarrow}$.
    Observe that $w(E(r_{I\setminus\{j\}}^{\downarrow}, u^{\downarrow})) = \sum_{i\in I\setminus \set{j}} w(E(r_i^{\downarrow}, u^{\downarrow}))$.
    Our goal is to find a child $r_{i^*}$ of $r$ such that $w(E(r_{i^*}^{\downarrow}, u^{\downarrow}))>\mintcut_G(u^{\downarrow})/2$ if it exists.
    To do so, split the set $[k]$ into two halves $I_1,I_2$, and recurse on $\arg\max_{J\in \set{I_1\setminus\{j\},I_2\setminus\{j\}}} w(E(r_J^{\downarrow}, u^{\downarrow}))$.
    This binary search terminates after at most $\log k$ iterations and returns a vertex $r_{i^*}$ such that $u$ is interested in $r_{i^*}$ if it exists.
    We again use the fact that $w(E(r_{I\setminus\{j\}}^{\downarrow}, u^{\downarrow}))$ can be found by querying the cuts $(\cup_{i\in I\setminus\{j\}} r_i^{\downarrow}) \cup u^{\downarrow},(\cup_{i\in I\setminus\{j\}} r_i^{\downarrow}),u^{\downarrow}$, which are unions of at most three contiguous sets in the post-order traversal of $T$, and therefore can be queried in $\polylog(n)$ time using the tree cut-query oracle.
    Therefore, the total time complexity of this process is $\polylog(n)$.

    Finally, we find the last vertex $\alpha\in P$ in which $u$ is interested, by performing a binary search on $P$.
    By \Cref{observation:interested-edges-structure}, if $u$ is interested in some vertex $x\in P$, then $u$ is also interested in every vertex above $x$ in $P$.
    Notice that $u$ is not a descendant of any vertex in $P$.
    Therefore, it is possible to find $w(E(x^{\downarrow}, u^{\downarrow}))$ by querying the cuts $x^{\downarrow}\cup u^{\downarrow},x^{\downarrow},u^{\downarrow}$ using the tree cut-query oracle.
    This immediately yields a binary search procedure to find the last vertex $\alpha\in P$ in which $u$ is interested using $\polylog(n)$ time.
    Finally, repeat the same process in the tree rooted at $\alpha$ to find all vertices in which $u$ is interested in this tree.

    To bound the time complexity, note that every leaf to root path in $T$ is decomposed into $O(\log n)$ paths in the heavy-light decomposition, and that throughout the algorithm we perform at most two binary searches on each of these paths, which yields a total time complexity of $\polylog(n)$.
\end{proof}

\subsection{Tree Packing Stability}
\label{sec:tree-packing-stability}
In this section we prove \Cref{lemma:tree-packing-stability}.
We note that a similar result was previously shown in \cite{CQX19}, however their goal was to obtain a better bound on the number of approximate minimum cuts rather than to show stability under edge updates.
We first provide a few results that will be used in the proof.
\begin{lemma}[\cite{PST91}]
    \label{lemma:tree-packing-procedure}
    There exists a deterministic algorithm that, given a weighted graph $G$ on $n$ vertices and $m$ edges with minimum cut value $\lambda_G$, constructs a weighted tree packing $\T$ of $G$ with value at least $(1-\epsilon)\lambda_G/2$ and at most $O(\epsilon^{-2}\lambda_G)$ distinct trees in time $\tO(\epsilon^{-2} m \lambda_G)$.
\end{lemma}
\begin{lemma}[Lemma 2.3 of \cite{Kar00}]
    \label{lemma:structural-packing-2-respecting}
    Given any weighted tree packing $\T$ of a graph $G$ with value $\beta\lambda_G$ and any cut of value $\alpha\lambda_G$, at least a $\tfrac{1}{2}(3-\tfrac{\alpha}{\beta})$ fraction of the trees in $\T$ $2$-respect the cut.
\end{lemma}

Observe that by \Cref{lemma:tree-packing-procedure}, the complexity of constructing the tree packing is proportional to $m\cdot \lambda_G$.
Therefore, in order to keep the processing time to $\tO(m)$, we need to first lower the value of $\lambda_G$ to $\polylog (n)$.
To do so, we construct a skeleton graph $H$ of $G$ which approximately preserves all near minimum cuts of $G$ and has minimum cut value $O(\log n)$.
\begin{theorem}[Theorem 7.8 of \cite{HLRW24}]
    \label{lemma:skeleton-graph-construction}
    There exists a deterministic algorithm that, given a weighted graph $G$ on $n$ vertices and $m$ edges with minimum cut value $\lambda_G$ and parameter $\epsilon\in [0,1]$, constructs in time $\tO(\epsilon^{-O(1)}m)$ an unweighted graph $H$ and weight parameter $W'=(\epsilon/\log n)^{O(1)} \lambda_G$ such that,
    \begin{enumerate}
        \item For any cut $C\subseteq V$, such that $\mintcut_G(C) \le \alpha \lambda_G$ for a small enough constant $\alpha>1$, $W'\cdot \mintcut_H(C)\le (1+\epsilon)\mintcut_G(C)$, and
        \item For any cut $C\subseteq V$, $ W'\cdot \mintcut_H(C) \ge (1-\epsilon)\mintcut_G(C)$.
    \end{enumerate}
\end{theorem}
It is straightforward to verify that the skeleton graph $H$ has minimum cut value $\polylog (n)$.
Note that the original theorem statement of \cite{HLRW24} only guarantees (1) for minimum cuts of $G$, however it is not hard to verify that the same proof also holds for cuts of value at most $\alpha \lambda_G$ for a small enough constant $\alpha>1$.
In addition, there exist randomized constructions of skeleton graphs that have better polylogarithmic factors, however we opt to use a single procedure even for our randomized algorithm for simplicity.

\begin{proof}[Proof of \Cref{lemma:tree-packing-stability}]
    Let $\alpha^{*} =\min\set{\alpha-1,1/100}$ where $\alpha$ is the parameter from (1) in \Cref{lemma:skeleton-graph-construction}.
    Begin by constructing a skeleton graph $H'$ using \Cref{lemma:skeleton-graph-construction} with $\epsilon=\alpha^*$.
    The number of edges in $H'$ is at most $\tO(m)$ since \Cref{lemma:skeleton-graph-construction} takes $\tO(m)$ time.
    In addition, observe that the minimum cut of $H'$ has $\polylog (n)$ edges.

    Applying \Cref{lemma:tree-packing-procedure} on $H'$ with $\epsilon=\alpha^*$ we get a tree packing $\T$ of $H'$ with value at least $(1-\epsilon)\lambda_{H'}/2\ge 2\lambda_{H'}/5$ and at most $O(\lambda_{H'})$ distinct trees in $\tO(m)$ time.
    Fix some $1+\alpha^*$-approximate minimum cut $C\subseteq V$ of $G$.
    By \Cref{lemma:skeleton-graph-construction}, $C$ is a $(1+\alpha^*)^2/(1-\alpha^*)\le (1+1/30)$ approximate minimum cut of $H'$.
    Hence, by \Cref{lemma:structural-packing-2-respecting}, at least a $\tfrac{1}{2}(3-\tfrac{1+1/30}{2/5})=\Omega(1)$ fraction of the trees in $\T$ $2$-respect $C$.
    Therefore, at least one tree in $\T$ $2$-respects $C$.
    
    To conclude the proof, notice that when $G$ undergoes at most $\eta \lambda_G$ unit weight edge updates, the value of every cut changes by at most $\eta \lambda_G$.
    Choosing $\eta=\alpha^*/2$ notice that the minimum cut of $G_t$ has value at most $(1+\alpha^*/2)\lambda_G$.
    Furthermore, any cut that achieves this value must have value at most $(1+\alpha^*)\lambda_G$ in $G$, and hence must be $2$-respected by some tree in $\T$.
    Therefore, every minimum cut of $G_t$ is $2$-respected by some tree in $\T$.
\end{proof}

\subsection{Dynamic tree cut-query oracle}
\label{sec:dynamic-cut-query-oracle}
In this section we prove \Cref{lemma:dynamic-T-cut-query-oracle}.
The data structure is based on the $2$D semigroup range searching data structure of \cite{Chazelle88}.
Throughout this section we use the semigroup $\R_{\ge 0}$ with the addition operation.
\begin{definition}[$2$D Semigroup Range Searching]
    Let $f$ be some function from points in $\R^2$ to a commutative semigroup $(G,+)$.
    A semigroup range searching data structure stores $m$ points $p_1,\ldots,p_m$ on the plane and supports the following query.
    Given an axes-aligned rectangle $R$, return $\sum_{i\in [m]} \indic{p_i\in R} f(p_i)$.
\end{definition}
\begin{lemma}[\cite{Chazelle88}]
    \label{lemma:chazelle-range-counting}
    There exists a $2$D semigroup range searching data structure that uses $\polylog m$ worst-case time to add/remove a point and $O(\log^2 m)$ time to answer a range query.
\end{lemma}
\begin{remark}
    The update time complexity of the data structure of \cite{Chazelle88} is amortized, however (as stated in the paper) it can be converted to worst-case using standard techniques.
\end{remark}
\begin{proof}[Proof of \Cref{lemma:dynamic-T-cut-query-oracle}]
    The proof mostly follows Lemma 6.11 in \cite{MN20}, except for performing queries of contiguous vertices in the post-order traversal.
    Order the vertices of $V$ according to a post-order traversal of the vertices $v_1,\ldots,v_n$.
    Then, initiate an empty semigroup range searching data structure $A$ of \Cref{lemma:chazelle-range-counting}, and for every edge $(v_i,v_j)$ place a point at coordinates $(i,j),(j,i)$ with its weight equal to the weight of the edge.
    Notice that this requires $\tO(|E|)$ time, and that every edge insertion/deletion can be performed by adding/removing a point in $A$ in $\polylog (n)$ time by \Cref{lemma:chazelle-range-counting}.
    We now explain how to answer all query types specified by \Cref{definition:T-cut-query-oracle}.
    We begin by showing how to answer every query that corresponds to a cut that is the union of at most $O(1)$ contiguous sets in the post-order traversal.
    
    In the case of a single contiguous set, denote the cut $S=\set{v_a,v_{a+1},\ldots,v_b}$.
    Observe that querying $A$ at $[p,q]\times[r,x]$ (when the intervals are disjoint), returns $w(E(\set{v_p,\ldots,v_q},\set{v_r,\ldots,v_x}))$.
    Therefore, summing the queries $[a,b]\times[1,a-1],[a,b]\times[b+1,n]$ returns exactly the sum of all edges with one endpoint in $S=\set{v_a,\ldots, v_b}$ and the other in $V\setminus S$.
    See \Cref{fig:contiguous-cuts} for an illustration of these queries.

    In the case of a union of two contiguous sets, denote the cut $S=\set{v_{a_1},v_{a_1+1},\ldots,v_{b_1}}\cup\set{v_{a_2},v_{a_2+1},\ldots,v_{b_2}}$.
    Assume without loss of generality that $b_1<a_2$.
    Similarly to the previous case, observe that summing the queries $[a_1,b_1]\times[1,a_1-1],[a_1,b_1]\times[b_1+1,a_2-1],[a_1,b_1]\times[b_2+1,n]$ returns $w(E(\set{v_{a_1},\ldots,v_{b_1}},V\setminus S))$.
    Therefore, the query can be answered by summing these queries and the symmetric queries for the second contiguous set.
    For an illustration of these queries, see \Cref{fig:contiguous-cuts}.
    Similarly, one can answer also queries that are the union of $O(1)$ contiguous sets by summing the appropriate queries to $A$.

\begin{figure}[htpb]
  \centering
  \scalebox{0.6}{
    \begin{tikzpicture}[
  rect/.style  = {fill=blue!15, draw=blue!70!black, line width=0.6pt},
  frame/.style = {draw=gray!55, dashed, line width=0.4pt},
  guide/.style = {draw=gray!40, dashed, line width=0.3pt},
  bar/.style   = {draw=orange!80!red, line width=1.2pt, line cap=round},
  axis/.style  = {draw=gray!65, -{Latex[length=4pt]}, line width=0.4pt},
  rlab/.style  = {text=orange!80!red, font=\bfseries\small},
]

\draw[frame] (0,0) rectangle (10,10);

\draw[guide] (3,0)  -- (3,10);
\draw[guide] (6,0)  -- (6,10);
\draw[guide] (0,7)  -- (10,7);
\draw[guide] (0,4)  -- (10,4);

\filldraw[rect] (0,4) rectangle (3,7);   %
\filldraw[rect] (6,4) rectangle (10,7);  %

\draw[bar] (3,10.5)  -- (6,10.5);
\draw[bar] (3,10.35) -- (3,10.65);
\draw[bar] (6,10.35) -- (6,10.65);
\node[rlab, above=1pt] at (3,10.5) {$a$};
\node[rlab, above=1pt] at (6,10.5) {$b$};

\draw[bar] (-0.7,4)  -- (-0.7,7);
\draw[bar] (-0.85,4) -- (-0.55,4);
\draw[bar] (-0.85,7) -- (-0.55,7);
\node[rlab, left=1pt] at (-0.7,7) {$a$};
\node[rlab, left=1pt] at (-0.7,4) {$b$};

\draw[axis] (0,-0.25)  -- (10.4,-0.25);
\draw[axis] (-0.25,10) -- (-0.25,-0.4);

\node[below=2pt, font=\small] at (0,-0.25)  {$1$};
\node[below=2pt, font=\small] at (10,-0.25) {$n$};
\node[anchor=east, font=\small] at (-0.3,10) {$1$};
\node[anchor=east, font=\small] at (-0.3,0)  {$n$};

\node[font=\small] at (5,-1.6) {(i)};

\draw[frame] (13,0) rectangle (23,10);

\draw[guide] (14.8,0) -- (14.8,10);
\draw[guide] (16,0)   -- (16,10);
\draw[guide] (18.5,0) -- (18.5,10);
\draw[guide] (20,0)   -- (20,10);

\draw[guide] (13,8.5) -- (23,8.5);
\draw[guide] (13,7)   -- (23,7);
\draw[guide] (13,4.5) -- (23,4.5);
\draw[guide] (13,3)   -- (23,3);

\filldraw[rect] (13,7)   rectangle (14.8,8.5);  %
\filldraw[rect] (16,7)   rectangle (18.5,8.5);  %
\filldraw[rect] (20,7)   rectangle (23,8.5);    %

\filldraw[rect] (13,3)   rectangle (14.8,4.5);
\filldraw[rect] (16,3)   rectangle (18.5,4.5);
\filldraw[rect] (20,3)   rectangle (23,4.5);

\draw[bar] (14.8,10.5)  -- (16,10.5);
\draw[bar] (14.8,10.35) -- (14.8,10.65);
\draw[bar] (16,10.35)   -- (16,10.65);
\node[rlab, above=1pt] at (14.8,10.5) {$a_1$};
\node[rlab, above=1pt] at (16,10.5)   {$b_1$};

\draw[bar] (18.5,10.5)  -- (20,10.5);
\draw[bar] (18.5,10.35) -- (18.5,10.65);
\draw[bar] (20,10.35)   -- (20,10.65);
\node[rlab, above=1pt] at (18.5,10.5) {$a_2$};
\node[rlab, above=1pt] at (20,10.5)   {$b_2$};

\draw[bar] (12.3,7)    -- (12.3,8.5);
\draw[bar] (12.15,7)   -- (12.45,7);
\draw[bar] (12.15,8.5) -- (12.45,8.5);
\node[rlab, left=1pt] at (12.3,8.5) {$a_1$};
\node[rlab, left=1pt] at (12.3,7)   {$b_1$};

\draw[bar] (12.3,3)    -- (12.3,4.5);
\draw[bar] (12.15,3)   -- (12.45,3);
\draw[bar] (12.15,4.5) -- (12.45,4.5);
\node[rlab, left=1pt] at (12.3,4.5) {$a_2$};
\node[rlab, left=1pt] at (12.3,3)   {$b_2$};

\draw[axis] (12.75,-0.25) -- (23.4,-0.25);
\draw[axis] (12.75,10)    -- (12.75,-0.4);

\node[below=2pt, font=\small] at (13,-0.25) {$1$};
\node[below=2pt, font=\small] at (23,-0.25) {$n$};
\node[anchor=east, font=\small] at (12.7,10) {$1$};
\node[anchor=east, font=\small] at (12.7,0)  {$n$};

\node[font=\small] at (18,-1.6) {(ii)};

\end{tikzpicture}
  }
  \caption{An illustration of contiguous cuts. (i) a single contiguous cut, (ii) a union of two contiguous cuts. The blue rectangles represent the queries to $A$ that are needed to answer the query.}
  \label{fig:contiguous-cuts}
\end{figure}
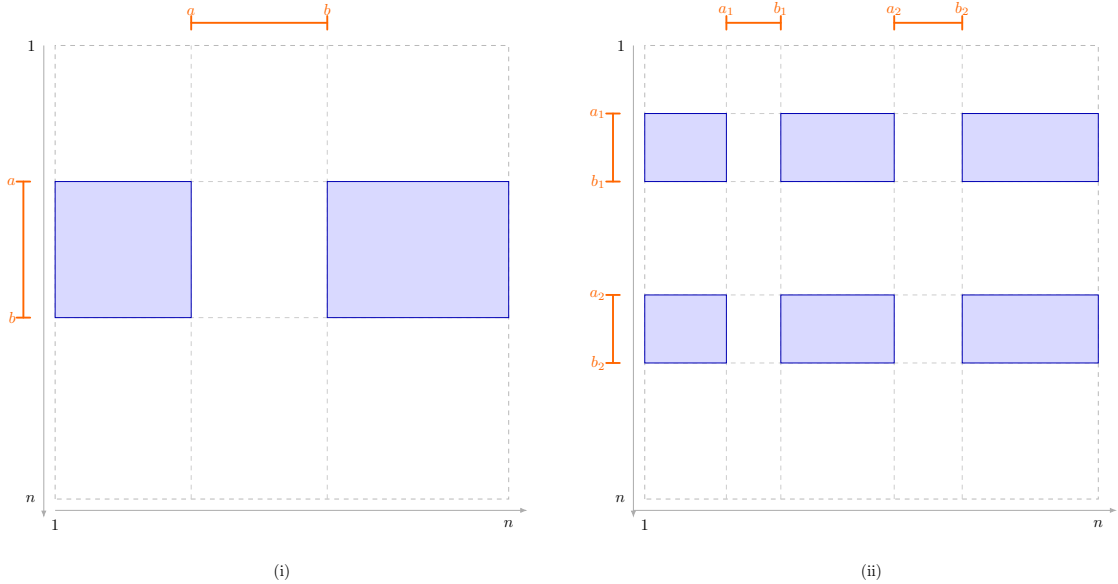

    Using the above queries, one can answer all queries of $2$-respecting cuts.
    Notice that every cut that $1$-respects $T$ is of the form $u^{\downarrow}$ for some $u\in V$.
    Furthermore, $u^{\downarrow}$ is a contiguous range in the ordering since it is formed by a post order traversal.
    Next, notice that $2$-respecting cuts are either of the form $u^{\downarrow}\cup w^{\downarrow}$ for some $u,w\in V$, or of the form $u^{\downarrow}\setminus w^{\downarrow}$ for some $u,w\in V$ such that $u$ is an ancestor of $w$ in $T$.
    Similarly to the previous case, these cuts are also formed by one or two contiguous ranges in the ordering, and therefore the procedure described above can find their value.
\end{proof}

\subsection{Putting It All Together}
\label{sec:putting-it-all-together}
In this section we prove \Cref{lemma:encapsulation-of-mn20}.
The proof is immediate given the following result of \cite{MN20} and \Cref{lemma:find-all-interesting-pairs}.
The result is a variant of Theorem 6.1 in \cite{MN20}, which is given the data structures described in the previous sections instead of computing them from scratch.
\begin{theorem}
    \label{theorem:mn20-2-respecting}
    There exists a deterministic algorithm that, given a spanning tree $T$ of a weighted graph $G$, a tree cut oracle corresponding to them, and the list of interested vertices, finds a minimum cut among those that $2$-respect $T$ in time $\tO(n)$.
\end{theorem}
\begin{proof}[Proof of \Cref{lemma:encapsulation-of-mn20}]
    Begin by finding all the interested vertices using \Cref{lemma:find-all-interesting-pairs} in $\tO(n)$ time.
    Then, use \Cref{theorem:mn20-2-respecting} to find a minimum cut among those that $2$-respect $T$ and return it.
    This concludes the proof.
\end{proof}

\section{Sublinear Update Time for Edge Connectivity}
\label{sec:sublinear-edge-connectivity}
In this section we prove \Cref{lemma:sublinear-edge-connectivity}, giving an algorithm that maintains the edge connectivity of a dynamic graph $G$ on $n$ vertices with $\tO(n/\delta_G)$ worst-case update time, where $\delta_G$ is the minimum degree of $G$.
We begin by introducing some tools that will be used throughout the section, and then we give an overview of our algorithm and its analysis.
We first introduce the following data structure from \cite{HKMR25}.
\begin{definition}[Dynamic Cutset (DCS)]
    A \emph{Dynamic Cutset (DCS)} is a dynamic data structure that maintains a forest $F$ and a dynamic graph $G$. 
    It supports the following operations on $F$:
    \begin{itemize}
        \item $\insertF(e)$: Insert $e$ into $F$.
        \item $\deleteF(e)$: Delete $e$ from $F$.
    \end{itemize}
    It also supports the following operations on $G$:
    \begin{itemize}
        \item $\insertG(e)$: Insert $e$ into $G$.
        \item $\deleteG(e)$: Delete $e$ from $G$. If $e$ is in $F$, remove it from $F$.
        \item $\query(v)$: Return an edge $e$ in $G$ that has one endpoint in the component of $F$ containing $v$ and one endpoint elsewhere, or return $\bot$ if no such edge exists.
    \end{itemize}
\end{definition}
\begin{theorem}[Lemma 8 of \cite{HKMR25}]
    There exists a randomized DCS data structure with worst-case update time of $\polylog (n)$ and query time of $\polylog (n)$.
\end{theorem}
\begin{remark}
    There also exists a deterministic DCS data structure which requires an $n^{o(1)}$ overhead in the update and query time, which can be used to derandomize parts of our algorithm.
\end{remark}
We will also use the following easy result.
\begin{claim}
    \label{claim:retrieve-maximal-packing}
    There exists an algorithm that, given a DCS on a graph $G$ with forest $F$ such that $F$ has $r$ connected components, the list of component IDs, and every vertex $v$ knows the ID of the component it belongs to in $F$, constructs a maximal $k$-packing of forests of $G/F$ with $O(rk)$ edges in $\tO(rk)$ time.
\end{claim}
\begin{proof}
    The algorithm iteratively constructs a single spanning tree $T$ of $G/F$, deletes all edges of $T$ from $G/F$, and repeats this process $k$ times.
    It is easy to see that the result is a maximal $k$-packing of forests of $G/F$.
    We show that each iteration takes $\tO(r)$ time, and adds $O(r)$ edges to the sparsifier yielding our desired result.

    To construct a spanning tree of $G/F$, start with some $v\in V$.
    Then, call $\query(v)$ to find an edge $e=(u,w)$ and add it to $T$, where $u$ is in the same component as $v$ in $F$ and $w$ is in a different component.
    It is possible to determine the endpoints of $e$ in $T$ by checking which component of $F$ contains $w$, and which component contains $u$.
    This can be done in $O(1)$ time as each vertex knows the ID of the component it belongs to by the theorem statement, and the components of $F$ are the vertices of $G/F$.
    Then, merge the components containing $u,w$ in $F$ by calling $\insertF((u,w))$.
    Repeat this process until $\query(v)$ returns $\bot$, which means that no more vertices exist in the spanning tree containing the component of $v$.
    Then, iterate over the next component in $G/F$ and repeat the process until every component of $G/F$ has been processed.
    At the end of the process, every component in  $T$ has an outgoing edge or it was found to be an isolated vertex in $G/F$.
    
    The total time complexity of this process is $\tO(r)$ as it adds $O(r)$ edges to $T$ and each edge is found in $\polylog (n)$ time.
    To repeat the process $k$ times, delete the edges of $T$ from $G/F$ by calling $\deleteG(e)$ for each edge $e\in T$, which takes $\tO(r)$ time.
    This resets the data structure to its original state (minus the edges of $T$), and allows us to repeat the process for the next iteration.
    After the construction, restore the original state of the DCS by calling $\insertG(e)$ for each edge $e$ that was deleted from $G/F$.
\end{proof}
\subsection{Algorithm Overview}
We are now ready to describe the algorithm of \Cref{lemma:sublinear-edge-connectivity}.
Our algorithm uses the star-contraction procedure \cite{AEGLMN22} to construct a weak NMC sparsifier (see \Cref{definition:weak-nmc}).
Recall that a weak NMC sparsifier of a graph $G$ is a contraction $H$ of $G$, such that for every cut $C\subseteq V$ of value at most $\lambda_G+\tfrac{1}{10}\delta_G$ in $G$, no edge in $E(C,V\setminus C)$ is contracted in $H$ with constant probability.

We will use the $\tau$-star-contraction procedure, introduced by \cite{KK25}, which is a variant of the star-contraction procedure of \cite{AEGLMN22} that is defined for a degree threshold $\tau$.
Sample a set of \emph{center vertices} $R\subseteq V$ by including each vertex $v\in V$ independently with probability $p=O(\log n/\tau)$.
Let $V_{\ge\tau}=\set{v\in V\setminus R \mid d_G(v)\ge \tau}$ be the set of non-center vertices with degree at least $\tau$,
where throughout $d_G(v)$ is the degree of $v$ in $G$.
Then, for every $v\in V_{\ge\tau}$, uniformly sample a vertex $r\in N_G(v)\cap R$, where $N_G(v)$ is the neighborhood of $v$ in $G$, and contract the corresponding edge $(v,r)$, keeping parallel edges, which yields a contracted multigraph $G'$.
The main guarantee of this procedure is that every non-trivial near-minimum cut of $G$, i.e., not composed of a single vertex, with $\mintcut_G(C)\le \lambda_G+\tfrac{1}{10}\delta_G$ is preserved in $G'$ with constant probability.
The following theorem states this formally.
\begin{restatable}[\cite{AEGLMN22,KK25,KK26a}]{theorem}{starcontraction}
    \label{theorem:star-contraction}
    Let $G=(V,E)$ be an unweighted graph on $n$ vertices and let $C$ be a non-trivial cut with $\mintcut_G(C)\le \lambda_G+\tfrac{1}{10}\delta_G$.
    Then, $\tau$-star-contraction with $p = 800\log n/\tau$ yields a contracted graph $G'$ such that if $\tau \le \delta_G$,
    \begin{enumerate}
        \item with probability at least $1-1/n^4$ all vertices in $V\setminus R$ are contracted and $G'$ has at most $\tO(n/\tau)$ vertices, and
        \item no edge in $E(C,V\setminus C)$ is contracted with probability at least $2\cdot 3^{-13}$.
    \end{enumerate}
\end{restatable}
We note that previous works \cite{AEGLMN22,KK26a} only prove the theorem for exact minimum cuts, however the same proof extends for any cuts of value $\lambda_G+\epsilon\delta_G$ for a small enough $\epsilon>0$.
For completeness, we include the proof of \Cref{theorem:star-contraction} in \Cref{appendix:star-contraction}.
We will use the implementation of the $\tau$-star-contraction procedure of \cite{KK26a} for the dynamic setting.
The main insight of the implementation is that one can fix the center vertices $R$ at the beginning of the algorithm, and then simply (re)assign vertices to centers as new edges come in.
We use this insight to maintain a DCS on the contracted graph $G/F$ of the $\tau$-star-contraction, where $F$ is the forest of contracted edges.
\begin{claim}
    \label{claim:star-contraction-dcs}
    There exists an algorithm that, given a dynamic graph $G$ on $n$ vertices undergoing edge insertions and deletions, and a degree threshold $\tau$, maintains a DCS on $G$ with $F$ being the forest of contracted edges in the $\tau$-star-contraction of $G$ using worst-case update time of $\poly(\log n)$.
    The algorithm is randomized and succeeds with high probability.
\end{claim}
\begin{proof}
    Sample a set of center vertices $R\subseteq V$ by including each vertex $v\in V$ independently with probability $p=O(\log n/\tau)$.
    In addition, initialize a DCS on $G$ with $F$ being the empty forest.
    For each vertex $v\in V\setminus R$, maintain a pointer to the center vertex $r\in R$ that $v$ is currently contracted to, or $\bot$ if $v$ is not contracted to any center vertex.
    Finally, for each $v\in V\setminus R$, maintain a heap of the neighbors of $v$ in $R$ sorted by priority as described later.

    Whenever an edge $(u,v)$ is inserted into $G$, insert it into the DCS by calling $\insertG((u,v))$.
    In addition, if $u\in R$ and $v\notin R$ (or vice versa), insert $u$ into the heap of $v$ (or insert $v$ into the heap of $u$) with a priority drawn uniformly at random from $[0,1]$.
    Then, if the element with minimum priority in the heap of $v$ (or $u$) has changed, call $\deleteF((v,r))$ (or $\deleteF((u,r))$) for the previous center vertex $r$ that $v$ (or $u$) was contracted to, and call $\insertF((v,r'))$ (or $\insertF((u,r'))$) for the new center vertex $r'$ with minimum priority in the heap of $v$ (or $u$).
    The procedure for edge deletion is analogous.

    It is easy to verify that this procedure samples a center vertex for $v$ uniformly at random from its neighborhood in $R$.
    Finally, note that the time to process an edge update is $\polylog(n)$, since it is only to update the heap, and do at most $3$ queries to the DCS data structure.
\end{proof}

Our algorithm maintains a DCS of a $\tau$-star-contraction of the input graph $G$ for $\tau=2^i$ for $i\in [\log n]$ using \Cref{claim:star-contraction-dcs}, and in parallel maintains the minimum degree $\delta_G$ of $G$ using a heap.
The main part of our algorithm is divided into epochs, each of length $c\cdot \lambda_G/50$ updates, for the appropriate constant $c$ in \Cref{theorem:dynamic-mn20}.
We will use the following easy claim regarding the preservation of cuts within each epoch.
\begin{claim}
    \label{claim:cut-preservation-epoch}
    Let $H_i$ be the graph obtained by constructing a maximal $2\lambda_{G_i}$-packing of forests of a $\tau$-star-contraction of $G_i$, and let $H_{i+t}$ be the graph obtained by reflecting each edge update in $G_i$ to $H_i$ for $t$ updates.
    If $C\subseteq V$ is a non-trivial minimum cut of $G_{i+t}$ and $t\le \lambda_{G_i}/20$ then $C$ is a minimum cut of $H_{i+t}$ with constant probability.
\end{claim}
\begin{proof}
    Denote by $G_i'$ the $\tau$-star-contraction of $G_i$ that is used to construct $H_i$.
    Notice that by \Cref{lemma:might-as-well-nmc} while $t\le \lambda_{G_i}/20 \le \delta_{G_i}/20$, every non-trivial minimum cut of $G_{i+t}$ is a minimum cut of $G_{i+t}'$ with constant probability, where $G_{i+t}'$ is the NMC sparsifier with the same updates applied.
    In addition, by \Cref{lemma:might-as-well-ni}, every non-trivial minimum cut of $G_{i+t}'$ is a minimum cut of $H_{i+t}$ as long as $t\le \lambda_{H_i}/20=\lambda_{G_i}/20$.
    Combining the above, we get the desired result.
\end{proof}
We are now ready to prove \Cref{lemma:sublinear-edge-connectivity}.
\begin{proof}[Proof of \Cref{lemma:sublinear-edge-connectivity}]
    The proof follows by giving an algorithm that finds the edge connectivity with constant probability, to obtain a high probability guarantee apply the algorithm $O(\log n)$ times in parallel and return the minimum value found by any repetition.
    Throughout, maintain a DCS of a $\tau$-star-contraction of the input graph $G$ for $\tau=2^i$ for $i\in [\log n]$ using \Cref{claim:star-contraction-dcs}.
    In parallel, maintain the minimum degree $\delta_G$ of $G$ using a heap.
    The algorithm is divided into epochs, each of length $L\eqdef c\cdot \lambda_G/50$ updates, for the $\lambda_G$ at the beginning of the epoch and the constant $c$ in \Cref{theorem:dynamic-mn20}.
    If at the beginning of an epoch $\delta_G \le \Delta/2$ then the algorithm does nothing further during the epoch.
    This implies that when $\delta\le \Delta/2$, the update time of the algorithm is the same as the update time of the DCS, which is $\polylog(n) = \tO(\min \left\{ n/\Delta,n/\delta_G \right\})$.
    For the rest of the proof, assume that $\delta_G > \Delta/2$ at the beginning of the epoch and hence it is sufficient to show that the algorithm update and query time is $\tO(n/\delta_G)=\tO(\min\left\{ n/\delta_G,n/\Delta \right\})$.

    In addition, if $L=O(\log n)$, then $\lambda_G=O(\log n)$, and we can query the edge connectivity in $O(n/\delta)$ time by recovering a $\lambda_G$-packing of forests of a $2^{\lfloor \log_2 \delta_G \rfloor}$-star-contraction of $G$ using \Cref{claim:retrieve-maximal-packing}, and then applying some static $\tO(m)$-time minimum-cut algorithm, e.g. \cite{Kar00,GMW20,MN20}, on it.
    Therefore, for the rest of the proof we assume that $L=\Omega(\log n)$.
    Let $T' \in [T- L/4, T]$, where $c<1$ is the small constant in \Cref{theorem:dynamic-mn20}.
    Notice that $\lambda_{G_T'}\in (1\pm 1/200) \lambda_{G_T}$ since there are at most $L/4=c\cdot \lambda_{G_T}/200$ updates between $T'$ and $T$. 

    Assume that at the beginning of each epoch, we have the data structure of \Cref{theorem:dynamic-mn20} constructed on a maximal $2\lambda_{G_{T'}}$-packing of forests of a $\tau$-star-contraction of $G_{T'}$, for $\tau\ge \delta_{G_{T'}}/2$.
    Observe that the age of the data structure at the beginning of the epoch is at most $L/4=c\cdot \lambda_{G_{T}}/200\le c\lambda_{G_T'}/150$.
    Since the length of the epoch is $L=c\cdot \lambda_{G_T}/50\le c\cdot \lambda_{G_T'}/40$, the data structure remains valid for the entire epoch by \Cref{theorem:dynamic-mn20}.
    In addition, by \Cref{claim:cut-preservation-epoch}, while we are in the current epoch, every non-trivial minimum cut of $G_{T'+t}$ is a minimum cut of $H_{T'+t}$ with constant probability.
    Hence, to answer edge connectivity queries for $G_{T'+t}$, return the minimum between the value returned by \Cref{theorem:dynamic-mn20} and the minimum degree of $G_{T'+t}$.

    We now describe the data structure construction at the end of each epoch.
    Denote the time when there are only $L/8$ updates left in the current epoch by $T^*$.
    Let $\tau^* = 2^{\lfloor \log_2 \delta_{G_{T^*}} \rfloor}$, and notice that $\tau^*\ge \delta_{G_{T^*}}/2$.
    Construct a maximal $2\lambda_{G_{T^*}}$-packing of forests of a $\tau^*$-star-contraction of $G_{T^*}$ using \Cref{claim:retrieve-maximal-packing}, and build a data structure of \Cref{theorem:dynamic-mn20} on it in the background.
    Notice that the total construction time is $\tO(n\lambda_{G_{T^*}}/\delta_{G_{T^*}})$.
    Split the construction into $L/16$ updates, such that after each update the algorithm spends $\tO(n/\delta_{G_{T^*}})$ time on the construction.
    During the last $L/16$ updates apply all the updates that were performed in $G$ since time $T^*$ to the data structure.
    Hence, the worst-case update time of the algorithm is $\tO(n/\delta_{G})$.
    Finally, note that during the construction of the maximal $2\lambda_{G_{T^*}}$-packing of forests, the DCS of the $\tau^*$-star-contraction of $G$ must stay static.
    To do so, we queue all edge updates, and apply them to the DCS during the construction of the \Cref{theorem:dynamic-mn20} data structure.

    Notice that at the beginning of the next epoch, denoted by $T''=T+L$, there exists a data structure of \Cref{theorem:dynamic-mn20} on a maximal $2\lambda_{G_{T^*}}$-packing of forests of a $\tau^*$-star-contraction of $G_{T^*}$.
    Furthermore, the age of the data structure at the beginning of the next epoch is at most $L/8=c\cdot \lambda_{G_{T^*}}/400 \le c\cdot \lambda_{G_{T''}}/200$ since $\lambda_{G_{T^*}}\in (1\pm 1/200) \lambda_{G_{T''}}$.
    Hence, it satisfies the assumption for the next epoch, and thus we can apply the same procedure for each epoch.

    Finally, the query time is $\tO(n/\delta_G)$ by using the data structure of \Cref{theorem:dynamic-mn20} on a graph with $\tO(n/\delta_G)$ vertices, and checking the minimum degree of $G$ using the heap.
\end{proof}

\subsection*{Acknowledgments}
In the writing of this paper, the authors used GitHub Copilot for in-line writing suggestions. 
In addition, the authors used Claude for feedback on several sections, particularly on the clarity and correctness of the English prose and the mathematical content. 
All proof ideas and technical content were developed by the authors, who maintain full responsibility for the content of the paper.

\bibliographystyle{alpha}
\bibliography{refs}

\newcommand{\etalchar}[1]{$^{#1}$}
\begin{thebibliography}{HKMR25}

\bibitem[AD21]{AD21}
Sepehr Assadi and Aditi Dudeja.
\newblock A simple semi-streaming algorithm for global minimum cuts.
\newblock In {\em 4th Symposium on Simplicity in Algorithms, {SOSA} 2021},
  pages 172--180. {SIAM}, 2021.

\bibitem[ADK{\etalchar{+}}16]{ADK+16}
Ittai Abraham, David Durfee, Ioannis Koutis, Sebastian Krinninger, and Richard
  Peng.
\newblock On fully dynamic graph sparsifiers.
\newblock In {\em {IEEE} 57th Annual Symposium on Foundations of Computer
  Science, {FOCS} 2016}, pages 335--344. {IEEE} Computer Society, 2016.

\bibitem[AEG{\etalchar{+}}22]{AEGLMN22}
Simon Apers, Yuval Efron, Pawel Gawrychowski, Troy Lee, Sagnik Mukhopadhyay,
  and Danupon Nanongkai.
\newblock Cut query algorithms with star contraction.
\newblock In {\em 63rd {IEEE} Annual Symposium on Foundations of Computer
  Science, {FOCS} 2022}, pages 507--518. {IEEE}, 2022.

\bibitem[AL21]{AL21}
Simon Apers and Troy Lee.
\newblock Quantum complexity of minimum cut.
\newblock In {\em 36th Computational Complexity Conference, {CCC} 2021}, volume
  200 of {\em LIPIcs}, pages 28:1--28:33. Schloss Dagstuhl - Leibniz-Zentrum
  f{\"{u}}r Informatik, 2021.

\bibitem[CGL{\etalchar{+}}20]{CGLNPS20}
Julia Chuzhoy, Yu~Gao, Jason Li, Danupon Nanongkai, Richard Peng, and
  Thatchaphol Saranurak.
\newblock A deterministic algorithm for balanced cut with applications to
  dynamic connectivity, flows, and beyond.
\newblock In {\em 61st {IEEE} Annual Symposium on Foundations of Computer
  Science, {FOCS} 2020}, pages 1158--1167. {IEEE}, 2020.

\bibitem[Cha88]{Chazelle88}
Bernard Chazelle.
\newblock A functional approach to data structures and its use in
  multidimensional searching.
\newblock {\em {SIAM} J. Comput.}, 17(3):427--462, 1988.

\bibitem[CQX19]{CQX19}
Chandra Chekuri, Kent Quanrud, and Chao Xu.
\newblock {LP} relaxation and tree packing for minimum k-cuts.
\newblock In {\em 2nd Symposium on Simplicity in Algorithms, {SOSA} 2019},
  OASIcs, pages 7:1--7:18. Schloss Dagstuhl - Leibniz-Zentrum f{\"{u}}r
  Informatik, 2019.

\bibitem[DEMN21]{DEM+21}
Michal Dory, Yuval Efron, Sagnik Mukhopadhyay, and Danupon Nanongkai.
\newblock Distributed weighted min-cut in nearly-optimal time.
\newblock In {\em {STOC} '21: 53rd Annual {ACM} {SIGACT} Symposium on Theory of
  Computing}, pages 1144--1153. {ACM}, 2021.

\bibitem[DHNS19]{DHN+19}
Mohit Daga, Monika Henzinger, Danupon Nanongkai, and Thatchaphol Saranurak.
\newblock Distributed edge connectivity in sublinear time.
\newblock In {\em Proceedings of the 51st Annual {ACM} {SIGACT} Symposium on
  Theory of Computing, {STOC} 2019}, pages 343--354. {ACM}, 2019.

\bibitem[dVC25]{dVC25}
Tijn de~Vos and Aleksander B.~G. Christiansen.
\newblock Tree-packing revisited: Faster fully dynamic min-cut and arboricity.
\newblock In {\em Proceedings of the 2025 Annual {ACM-SIAM} Symposium on
  Discrete Algorithms, {SODA} 2025}, pages 700--749. {SIAM}, 2025.

\bibitem[EGIN97]{EGIN97}
David Eppstein, Zvi Galil, Giuseppe~F. Italiano, and Amnon Nissenzweig.
\newblock Sparsification - a technique for speeding up dynamic graph
  algorithms.
\newblock {\em J. {ACM}}, 44(5):669--696, 1997.

\bibitem[EHL25]{ElHL25}
Antoine El{-}Hayek, Monika Henzinger, and Jason Li.
\newblock Fully dynamic approximate minimum cut in subpolynomial time per
  operation.
\newblock In {\em Proceedings of the 2025 Annual {ACM-SIAM} Symposium on
  Discrete Algorithms, {SODA} 2025}, pages 750--784. {SIAM}, 2025.

\bibitem[EHL26]{ElHL26}
Antoine El{-}Hayek, Monika Henzinger, and Jason Li.
\newblock Deterministic and exact fully-dynamic minimum cut of
  superpolylogarithmic size in subpolynomial time.
\newblock In {\em Proceedings of the 2026 Annual {ACM-SIAM} Symposium on
  Discrete Algorithms, {SODA} 2026}, pages 613--663. {SIAM}, 2026.

\bibitem[Fre85]{Fre85}
Greg~N. Frederickson.
\newblock Data structures for on-line updating of minimum spanning trees, with
  applications.
\newblock {\em {SIAM} J. Comput.}, 14(4):781--798, 1985.

\bibitem[GG18]{GG18}
Barbara Geissmann and Lukas Gianinazzi.
\newblock Parallel minimum cuts in near-linear work and low depth.
\newblock In {\em Proceedings of the 30th on Symposium on Parallelism in
  Algorithms and Architectures, {SPAA} 2018}, pages 1--11. {ACM}, 2018.

\bibitem[GHN{\etalchar{+}}23]{GHN+23}
Gramoz Goranci, Monika Henzinger, Danupon Nanongkai, Thatchaphol Saranurak,
  Mikkel Thorup, and Christian Wulff{-}Nilsen.
\newblock Fully dynamic exact edge connectivity in sublinear time.
\newblock In Nikhil Bansal and Viswanath Nagarajan, editors, {\em Proceedings
  of the 2023 {ACM-SIAM} Symposium on Discrete Algorithms, {SODA} 2023}, pages
  70--86. {SIAM}, 2023.

\bibitem[GI91]{GI91}
Zvi Galil and Giuseppe~F. Italiano.
\newblock Fully dynamic algorithms for edge-connectivity problems (extended
  abstract).
\newblock In {\em Proceedings of the 23rd Annual {ACM} Symposium on Theory of
  Computing, {STOC} 1991}, pages 317--327. {ACM}, 1991.

\bibitem[GK13]{GK13}
Mohsen Ghaffari and Fabian Kuhn.
\newblock Distributed minimum cut approximation.
\newblock In {\em Distributed Computing - 27th International Symposium, {DISC}
  2013}, volume 8205 of {\em Lecture Notes in Computer Science}, pages 1--15.
  Springer, 2013.

\bibitem[GKKT15]{GKKT15}
David Gibb, Bruce~M. Kapron, Valerie King, and Nolan Thorn.
\newblock Dynamic graph connectivity with improved worst case update time and
  sublinear space.
\newblock {\em CoRR}, abs/1509.06464, 2015.

\bibitem[GMW20]{GMW20}
Pawel Gawrychowski, Shay Mozes, and Oren Weimann.
\newblock Minimum cut in o(m log{\({^2}\)} n) time.
\newblock In {\em 47th International Colloquium on Automata, Languages, and
  Programming, {ICALP} 2020}, volume 168 of {\em LIPIcs}, pages 57:1--57:15.
  Schloss Dagstuhl - Leibniz-Zentrum f{"{u}}r Informatik, 2020.

\bibitem[GNT20]{GNT20}
Mohsen Ghaffari, Krzysztof Nowicki, and Mikkel Thorup.
\newblock Faster algorithms for edge connectivity via random 2-out
  contractions.
\newblock In {\em Proceedings of the 2020 {ACM-SIAM} Symposium on Discrete
  Algorithms, {SODA} 2020}, pages 1260--1279. {SIAM}, 2020.

\bibitem[HdLT01]{HdLT01}
Jacob Holm, Kristian de~Lichtenberg, and Mikkel Thorup.
\newblock Poly-logarithmic deterministic fully-dynamic algorithms for
  connectivity, minimum spanning tree, 2-edge, and biconnectivity.
\newblock {\em J. {ACM}}, 48(4):723--760, 2001.

\bibitem[HKMR25]{HKMR25}
Monika Henzinger, Evangelos Kosinas, Robin M{\"{u}}nk, and Harald R{\"{a}}cke.
\newblock Efficient contractions of dynamic graphs - with applications.
\newblock In {\em 33rd Annual European Symposium on Algorithms, {ESA} 2025},
  volume 351 of {\em LIPIcs}, pages 36:1--36:14. Schloss Dagstuhl -
  Leibniz-Zentrum f{\"{u}}r Informatik, 2025.

\bibitem[HLRW24]{HLRW24}
Monika Henzinger, Jason Li, Satish Rao, and Di~Wang.
\newblock Deterministic near-linear time minimum cut in weighted graphs.
\newblock In {\em Proceedings of the 2024 {ACM-SIAM} Symposium on Discrete
  Algorithms, {SODA} 2024}, pages 3089--3139. {SIAM}, 2024.

\bibitem[HRW17]{HRW17}
Monika Henzinger, Satish Rao, and Di~Wang.
\newblock Local flow partitioning for faster edge connectivity.
\newblock In {\em Proceedings of the Twenty-Eighth Annual {ACM-SIAM} Symposium
  on Discrete Algorithms, {SODA} 2017}, pages 1919--1938. {SIAM}, 2017.

\bibitem[JST24]{JST24}
Wenyu Jin, Xiaorui Sun, and Mikkel Thorup.
\newblock Fully dynamic min-cut of superconstant size in subpolynomial time.
\newblock In {\em Proceedings of the 2024 {ACM-SIAM} Symposium on Discrete
  Algorithms, {SODA} 2024}, pages 2999--3026. {SIAM}, 2024.

\bibitem[Kar00]{Kar00}
David~R. Karger.
\newblock Minimum cuts in near-linear time.
\newblock {\em J. {ACM}}, 47(1):46--76, 2000.

\bibitem[Ken26]{Kenneth-Mordoch26}
Yotam Kenneth{-}Mordoch.
\newblock Faster pseudo-deterministic minimum cut.
\newblock {\em CoRR}, abs/2602.14550, 2026.

\bibitem[KK25]{KK25}
Yotam Kenneth{-}Mordoch and Robert Krauthgamer.
\newblock Cut-query algorithms with few rounds.
\newblock In {\em 33rd Annual European Symposium on Algorithms, {ESA} 2025},
  LIPIcs, pages 100:1--100:14. Schloss Dagstuhl - Leibniz-Zentrum f{\"{u}}r
  Informatik, 2025.

\bibitem[KK26a]{KK26b}
Yotam Kenneth{-}Mordoch and Robert Krauthgamer.
\newblock Faster all-pairs minimum cut: Bypassing exact max-flow.
\newblock In {\em Proceedings of the 58th Annual {ACM} Symposium on Theory of
  Computing, {STOC} 2026}, pages 1254--1265. {ACM}, 2026.

\bibitem[KK26b]{KK26a}
Yotam Kenneth{-}Mordoch and Robert Krauthgamer.
\newblock Simple algorithms for fully dynamic edge connectivity.
\newblock In {\em 2026 Symposium on Simplicity in Algorithms, {SOSA} 2026},
  pages 394--403. {SIAM}, 2026.

\bibitem[KKM13]{KKM13}
Bruce~M. Kapron, Valerie King, and Ben Mountjoy.
\newblock Dynamic graph connectivity in polylogarithmic worst case time.
\newblock In {\em Proceedings of the Twenty-Fourth Annual {ACM-SIAM} Symposium
  on Discrete Algorithms, {SODA} 2013}, pages 1131--1142. {SIAM}, 2013.

\bibitem[KS96]{KS96}
David~R. Karger and Clifford Stein.
\newblock A new approach to the minimum cut problem.
\newblock {\em J. {ACM}}, 43(4):601--640, 1996.

\bibitem[KT19]{KT19}
Ken{-}ichi Kawarabayashi and Mikkel Thorup.
\newblock Deterministic edge connectivity in near-linear time.
\newblock {\em J. {ACM}}, 66(1):4:1--4:50, 2019.

\bibitem[MN20]{MN20}
Sagnik Mukhopadhyay and Danupon Nanongkai.
\newblock Weighted min-cut: sequential, cut-query, and streaming algorithms.
\newblock In {\em Proceedings of the 52nd Annual {ACM} {SIGACT} Symposium on
  Theory of Computing, {STOC} 2020}, pages 496--509. {ACM}, 2020.

\bibitem[NI92]{NI92}
Hiroshi Nagamochi and Toshihide Ibaraki.
\newblock A linear-time algorithm for finding a sparse k-connected spanning
  subgraph of a k-connected graph.
\newblock {\em Algorithmica}, 7(5{\&}6):583--596, 1992.

\bibitem[PST91]{PST91}
Serge~A. Plotkin, David~B. Shmoys, and {\'{E}}va Tardos.
\newblock Fast approximation algorithms for fractional packing and covering
  problems.
\newblock In {\em 32nd Annual Symposium on Foundations of Computer Science, San
  Juan, Puerto Rico}, pages 495--504. {IEEE} Computer Society, 1991.

\bibitem[RSW18]{RSW18}
Aviad Rubinstein, Tselil Schramm, and S.~Matthew Weinberg.
\newblock Computing exact minimum cuts without knowing the graph.
\newblock In {\em 9th Innovations in Theoretical Computer Science Conference,
  {ITCS} 2018}, volume~94 of {\em LIPIcs}, pages 39:1--39:16. Schloss Dagstuhl
  - Leibniz-Zentrum f{\"{u}}r Informatik, 2018.

\bibitem[Sar21]{Sar21}
Thatchaphol Saranurak.
\newblock A simple deterministic algorithm for edge connectivity.
\newblock In Hung~Viet Le and Valerie King, editors, {\em 4th Symposium on
  Simplicity in Algorithms, {SOSA} 2021}, pages 80--85. {SIAM}, 2021.

\bibitem[ST83]{ST83}
Daniel~Dominic Sleator and Robert~Endre Tarjan.
\newblock A data structure for dynamic trees.
\newblock {\em J. Comput. Syst. Sci.}, 26(3):362--391, 1983.

\bibitem[SW97]{SW97}
Mechthild Stoer and Frank Wagner.
\newblock A simple min-cut algorithm.
\newblock {\em J. {ACM}}, 44(4):585--591, 1997.

\bibitem[Tho07]{Tho07}
Mikkel Thorup.
\newblock Fully-dynamic min-cut.
\newblock {\em Comb.}, 27(1):91--127, 2007.

\bibitem[TK00]{TK00}
Mikkel Thorup and David~R. Karger.
\newblock Dynamic graph algorithms with applications.
\newblock In {\em Algorithm Theory - {SWAT} 2000, 7th Scandinavian Workshop on
  Algorithm Theory}, volume 1851 of {\em Lecture Notes in Computer Science},
  pages 1--9. Springer, 2000.

\end{thebibliography}
\appendix
\section{Star Contraction Proof}
\label{appendix:star-contraction}
In this section we prove a strengthened version of the $\tau$-star contraction of \cite{AEGLMN22,KK25} that preserves every non-trivial cut with value at most $\mintcut_G(C)\le \lambda_G+(1/10)\delta_G$.
We begin with restating the theorem for convenience.
\starcontraction*
The proof requires the following useful proposition from \cite{AEGLMN22}, and the following version of the Chernoff bound.
\begin{proposition}[Proposition 4.1 in \cite{AEGLMN22}]
    \label{proposition:probability-optimization}
    Let $n$ be a positive integer, $a \in [0,1)$, and $b\ge 1$. 
    Define
\begin{align*}
    F(a,b) 
    = 
    \min_{x\in \R^n}
    \quad
     &\prod_{i\in [n]} (1-x_i)
    \\
     \text{subject to}
    \quad
    &\sum_{i\in [n]} x_i = b,
    \\
    & \forall i\in [n]: 0\le x_i \le a.
\end{align*}
\end{proposition}
\begin{theorem}\label[theorem]{theorem:chernoff}
  Let $X_1,\ldots,X_m\in [0,a]$ be independent random variables.
    For any $\delta \in [0,1]$ and $\mu \geq \mathbb{E}\left[\sum_{i=1}^{m}X_i\right]$, we have
        \begin{align}
            \nonumber
            \mathbb{P}\left[
                \left|
                    \sum_{i=1}^{m}X_i - \mathbb{E}\left[\sum_{i=1}^{m}X_i\right]
                \right|
                \geq \delta \mu
            \right]
            \leq
            2\exp
            \left(
                -\frac{\delta^2\mu}{3a}
            \right)
            .
        \end{align}
\end{theorem}
\begin{proof}[Proof of \Cref{theorem:star-contraction}]
    The proof follows the same lines as  \cite{AEGLMN22,KK25,KK26a}, but we include it here for completeness.
    Throughout the proof we assume that $\delta_G \ge \tau$.
    Fix some cut $C \subseteq V$ of $G$, such that $n-1>|C|>1$, and let $\C\subseteq E$ be the set of edges crossing $C$.
    In addition, let $c(v)$ be the number of edges in $\C$ incident to $v$.
    We begin by arguing the correctness of the algorithm.
    Since the algorithm only performs edge contractions, it is clear that it can only increase the edge connectivity of the graph.
    Therefore, it remains to show that it does not contract any edge in $\C$ with constant probability.

    Observe that the probability that a contraction of a vertex $v\in H\setminus R$ hits an edge in $\C$ is equal to $c_R(v)/\deg_R(v)$, where $c_R(v) = |\C \cap E(v,R)|$, $d_R(v)=|E(v,R)|$ and recalling $H=\set{v\in V \mid \deg(v)\ge \tau}$.
    Therefore,
    \begin{equation*}
        \Probability{\text{$\C$ is not contracted} \mid R}
        = \prod_{v\in H\setminus R} \left(1-\frac{c_R(v)}{\deg_R(v)}\right)
        .
    \end{equation*}
    We say that a set $R$ is \emph{good} if $\sum_{v\in H\setminus R} \frac{c_R(v)}{d_R(v)} \le 8$ and $\max_{v\in H\setminus R} \frac{c_R(v)}{d_R(v)} \le 2/3$.
    Notice that if $R$ is good then the probability that $\C$ is not contracted is bounded by,
    \begin{align*}
        \Probability{\text{$\C$ is not contracted}}
        &\ge \Probability{\text{$\C$ is not contracted} \mid R \text{ is good}} \cdot \Probability{R \text{ is good}}
        \\
        &\ge \left( 1-\frac{2}{3} \right)^{\lceil 8/(2/3)\rceil} 
        \cdot \Probability{R \text{ is good}}
        = 3^{-12} \cdot \Probability{R \text{ is good}},
    \end{align*}
    where the first inequality is from the law of total probability and the second is from \Cref{proposition:probability-optimization}.
    Therefore, to prove the correctness guarantee of the theorem it suffices to show that $\Probability{R \text{ is good}} \ge 1/2$.

    Throughout we will use the following claim.
    \begin{claim}[Proposition 4.5 in \cite{AEGLMN22}]
        \label{claim:expectation-of-ratio}
    Let $G=(V,E)$ be a simple $n$-vertex graph and let $C\subseteq V$ be some cut.
    Choose a set $R$ by putting each vertex of $V$ into $R$ independently with some probability $p$.
    Then, for any $v\in V$,
    \begin{equation*}
        \Exp{c_R(v)/d_R(v) \mid d_R(v) >0}{} = \frac{c(v)}{\deg(v)}
        ,
    \end{equation*}
    where the expectation is taken over the choice of $R$.
    \end{claim}
    We first bound the probability that $\sum_{v\in H\setminus R} c_R(v)/d_R(v) \ge 8$.
    Then notice that,
    \begin{equation*}
        \sum_{v\in H\setminus R} \frac{c(v)}{\deg(v)}
        \le \sum_{v\in H\setminus R} \frac{c(v)}{\delta_G}
        \le \frac{2|\C|}{\delta_G}
        \le \frac{22}{10}
        ,     
    \end{equation*}
    where the last inequality is from $|\C|\le \lambda + (1/10)\cdot\delta_G\le (11/10)\delta_G$.
    By the linearity of expectation we have that,
    \begin{equation*}
        \Exp{\sum_{v\in H\setminus R} \frac{c_R(v)}{d_R(v)}}{R}
        \le \sum_{v\in H\setminus R} \Exp{\frac{c_R(v)}{d_R(v)} \mid d_R(v) > 0}{R}
        \le \sum_{v\in H\setminus R} \frac{c(v)}{\deg(v)}
        \le \frac{11}{5}
        ,
    \end{equation*}
    where the first inequality is from $\Exp{c_R(v)/d_R(v) \mid d_R(v) >0}{} \ge \Exp{\frac{c_R(v)}{d_R(v)}}{R}$, and the second is from \Cref{claim:expectation-of-ratio}.
    Therefore, by Markov's inequality we have,
    \begin{equation*}
        \Probability{\sum_{v\in H\setminus R : c_R(v)>0} \frac{c_R(v)}{d_R(v)} \ge 8}
        \le \frac{1}{3}
        .
    \end{equation*}
    For the rest of the proof assume that this event does not occur.

    We now turn to bound the probability that $\max_{v\in H\setminus R} c_R(v)/d_R(v) > 2/3$.
    We start by showing that $c(v)/\deg(v) \le \tfrac{11}{20}$.
    Observe that for any $v\in C$, $\mintcut_G(C\setminus \set{v}) = \mintcut_G(C)+\deg(v)-2c(v)$.
    Furthermore, $\mintcut_G(C)\le \lambda_G + (1/10)\delta_G$ and $\mintcut_G(C\setminus \set{v})\ge \lambda_G$.
    Hence,
    \begin{equation*}
        \lambda_G 
        \le \mintcut_G(C\setminus \set{v}) 
        = \mintcut_G(C)+\deg(v)-2c(v) 
        \le \lambda_G + (1/10)\delta_G +\deg(v)-2c(v)
        .
    \end{equation*}
    Rearranging the terms we obtain that $c(v) \le \tfrac{11}{20}\deg(v)$.
    Therefore, $\tfrac{c(v)}{\deg(v)} \le \tfrac{11}{20}\deg(v)/\deg(v) \le \tfrac{11}{20}$.
    We obtain a similar bound for $v\in V\setminus C$.

    Recall that we sample each vertex into $R$ with probability $p=800\cdot \log n/\tau$.
    Using the Chernoff bound, we obtain that $d_R(v)\ge 0.9 \deg(v)p$ with probability at least,
    \begin{equation*}
        \Probability{d_R(v) \ge 0.9 \deg(v)p} 
        \ge 1-\exp(-(1/10)^2 \deg(v)800\log n /(2\tau)) 
        \ge 1-1/n^4
        ,
    \end{equation*}
    where the last inequality is from $\deg(v)\ge \tau$.
    Therefore, $d_R(v)\ge 0.9 \deg(v)p$ for all $v\in H$ with probability at least $1-O(n^{-3})$.
    Assume this event holds for the rest of the proof.
    
    Notice that similarly to $d_R(v)$, we have $\Exp{c_R(v)}{R} = c(v)p$.
    Since $c(v)/\deg(v)\le 11/20$ we have that $\Exp{c_R(v)}{R} \le 11/20 \cdot p\deg(v)$.
    Hence, if $c_R(v)/d_R(v) > 2/3$ then $c_R(v) > 2d_R(v)/3\ge 2\cdot (9/10) \deg(v)p/3=(3/5)\cdot(20/11)\Exp{c_R(v)}{R} =(12/11)\Exp{c_R(v)}{R}$, where the last inequality is from $d_R(v) \ge 0.9 \deg(v)p$.
    Furthermore, since $d_R(v)\ge 0.9 \deg(v)p \ge 720 \log n$ we must have that $c_R(v) \ge 480 \log n$ for this event to occur.
    
    To bound the maximum, we split into two cases.
    The first, when $\Exp{c_R(v)}{R} \ge 400 \log n$.
    In this case by the Chernoff bound we have,
    \begin{equation*}
        \Probability{c_R(v) \ge (12/11)\cdot \Exp{c_R(v)}{R}}
        \le \exp(-(1/11)^2 \cdot 400 \log n / 3)
        \le n^{-1.1}
        .
    \end{equation*}
    The second case is when $\Exp{c_R(v)}{R} < 400 \log n$.
    Recall that $d_R(v)\ge 720 \log n$, and hence if we show that $c_R(v) \le 480 \log n$ then we can conclude that $c_R(v)/d_R(v) \le2/3$.
    In this case let $s =  480 \log n/\Exp{c_R(v)}{R} -1\ge 60\log n/\Exp{c_R(v)}{R}$.
    Then, by the Chernoff bound we have that,
    \begin{equation*}
        \Probability{c_R(v) \ge 480\log n}
        \le \exp(-s^2 \cdot \Exp{c_R(v)}{R} / 3)
        \le \exp(-60^2 \log n  \cdot (\log n/\Exp{c_R(v)}{R}) / 3)
        \le n^{-3}
        .
    \end{equation*}
    Combining the two cases we have that $\max_{v\in H\setminus R} c_R(v)/d_R(v)\le 2/3$ with probability at least $1-O(n^{-0.1})$.
    Hence, we can conclude that $R$ is good with probability at least $1/2$.

    To conclude the proof we now show the size guarantee.
    Partition the vertices of $G'$ into $R$ and $U=V\setminus R$.
    We have seen above that for every $v\in H \setminus R$ we have that $d_R(v)\ge 0.9 \deg(v)p\ge 720 \log n$ with probability at least $1-1/n^4$.
    Therefore, every vertex in $H\setminus R$ is contracted with high probability.
    Since every vertex in $V$ has degree at least $\tau$, we have that $H=V$.
    Finally, by a straightforward application of the Chernoff bound we have that $|R| = O(n/\tau \cdot \log n)$ with high probability.
\end{proof}

\section{Dynamic Maximum \texorpdfstring{$k$}{k}-Packing of Forests}
\label{sec:packing-of-forests}
In this section we show how to maintain a maximal $n$-packing of forests, proving \Cref{lemma:packing-of-forests}.
Such results were previously used in the literature (see e.g. \cite{ADK+16,KK26a}), and we include the proof for completeness.

Throughout we assume the existence of a minimum spanning forest data structure with the following guarantees.
\begin{definition}[Minimum Spanning Forest]
    A fully dynamic minimum spanning forest (MSF) data structure maintains a spanning forest of a dynamic graph $G$ undergoing edge insertions and deletions, and supports the following operations:
    \begin{itemize}
        \item \texttt{insert(e)}: Insert edge $e$ into $G$.
        \item \texttt{delete(e)}: Delete edge $e$ from $G$.
    \end{itemize}
    An edge insertion only changes the spanning forest if the inserted edge connects two components, and an edge deletion only changes the spanning forest if the deleted edge is in the spanning forest, in which case it finds a single replacement edge if one exists.
\end{definition}
\begin{theorem}
    There exist dynamic MSF data structures with the following guarantees:
    \begin{itemize}
        \item A deterministic data structure with $n^{o(1)}$ worst-case update time~\cite{CGLNPS20}.
        \item A deterministic data structure with $\polylog(n)$ amortized update time~\cite{HdLT01}.
        \item A randomized data structure with $\polylog(n)$ worst-case update time~\cite{KKM13,GKKT15}.
    \end{itemize}
\end{theorem}

\begin{proof}[Proof of \Cref{lemma:packing-of-forests}]
    The algorithm maintains $k$ dynamic MSF data structures, $T_1,T_2,\ldots,T_k$, in parallel one for each forest in the packing.
    The input for the MST $T_i$ is the graph $G_i = G \setminus \bigcup_{j=1}^{i-1} T_j$, i.e. the graph obtained by removing all edges in the previous trees.
    Given an edge insertion $e$, insert it into each of the $k$ MSTs, until it is successfully inserted into one of them.
    Given an edge deletion $e$, delete it from each of the $k$ data structures.
    If $e$ is in one of the trees, it is deleted from that tree and might be replaced by a single replacement edge, which is then deleted from all subsequent trees.
    It is easy to see that $\{T_1,\ldots,T_k\}$ is a maximal $k$-packing of forests of $G$.

    To conclude, we analyze the update time of the algorithm.
    Denote the update time of the MST data structure by $\tau$.
    Each inserted edge is added to at most $k$ trees, since once it is inserted into one of the $k$ data structures, it will not be inserted into any of the subsequent ones.
    For an edge deletion, notice that if the edge is in one of the trees, it is deleted from that tree and at most one replacement edge is inserted into that tree.
    This replacement edge is then deleted from all subsequent trees.
    Therefore, each tree sees at most a single edge deletion.
    Hence, the total update time for an edge deletion is at most $k \tau$.
    Plugging in the values of $\tau$ for the three MST data structures above gives the desired result.
\end{proof}

\section{Simple Deterministic Dynamic Minimum Cut Algorithm}
\label{sec:simple-deterministic}
In this section we give an alternative, self-contained proof of \Cref{theorem:deterministic-edge-connectivity}.
To prove \Cref{theorem:deterministic-edge-connectivity}, we use the following result on deterministic construction of a strong NMC sparsifier in near-linear time~\cite{KT19,HRW17}.
\begin{theorem}[\cite{KT19,HRW17}]
    \label{theorem:deterministic-nmc-sparsifier}
    There exists a deterministic algorithm that, given a simple graph $G$ on $n$ vertices and $m$ edges, constructs a strong NMC sparsifier of $G$ with $\tO(n/\delta_G)$ vertices and $\tO(m/\delta_G)$ edges in $\tO(m)$ time.
\end{theorem}

\begin{proof}[Proof of \Cref{theorem:deterministic-edge-connectivity}]
    The algorithm maintains a maximal $n$-packing of forests of the graph, and a heap of vertex degrees to track the minimum degree $\delta_G$.
    We divide the updates into epochs of $\delta_G/40$ updates.
    Assume that at the beginning of each epoch we are given a strong NMC sparsifier $H$ of a $2\delta_G$-NI sparsifier of $G$ whose age is less than $\delta_G/80$.
    Throughout the epoch, whenever an edge $e=(u,v)$ is inserted or deleted in $G$, insert/delete the corresponding edge in $H$ (if $u,v$ are not contracted in $H$).
    By combining \Cref{lemma:might-as-well-ni,lemma:might-as-well-nmc}, we have that every non-trivial minimum cut of $G_t$ is a minimum cut of $H_t$ for every $t$ in the epoch, where $H_t$ is the NMC sparsifier with the same updates applied.
    Therefore, to answer a query of the minimum cut of $G_t$, it suffices to run a static deterministic minimum cut algorithm on $H_t$ such as \cite{HLRW24} and return the minimum between it and the minimum degree of $G_t$.

    When there rest only $\delta_G/80$ updates left in the epoch, use the first $2\delta_G$ forests in the packing to build a $2\delta_G$-NI sparsifier of $G$, and then build a strong NMC sparsifier of the NI sparsifier.
    Notice that constructing the NI and NMC sparsifiers takes $\tO(n\delta_G)$ time by \Cref{theorem:deterministic-nmc-sparsifier}, and thus we can spend $\tO(n)$ time during $\delta_G/160$ updates to build the NMC sparsifier in the background.
    In the final $\delta_G/160$ updates of the epoch, we apply all the edge updates that arrived during the previous $\delta_G/160$ updates to the new NMC sparsifier, $\tO(\log n)$ time per update, and then switch to the new NMC sparsifier at the beginning of the next epoch.
    Observe that the age of the NMC sparsifier at the end of the epoch is at most $\delta_G/80$, and thus it is valid for the next epoch.
    The algorithm constructing the NMC sparsifier requires a static copy of the maximal $n$-packing of forests, which we obtain by queueing the updates to the  maximal $n$-packing of forests and applying them in the next $\delta_G/80$ updates after the NMC sparsifier is constructed.

    Finally, the overall time complexity of the algorithm is $n^{1+o(1)}$ worst-case update time or $\tO(n)$ amortized update time, depending on the choice of the algorithm for maintaining the maximal $n$-packing of forests by \Cref{lemma:packing-of-forests}, which concludes the proof.
\end{proof}

\section{Robustness to Edge Updates}
\label{sec:robustness-edge-updates}
In this section we prove that the minimum cut of a graph is preserved when we replace it by a maximal packing of forests, or an NMC sparsifier, as long as the number of updates is not too large compared to the minimum cut of the graph.
This proves \Cref{lemma:might-as-well-ni} and \Cref{lemma:might-as-well-nmc}.
\begin{proof}[Proof of \Cref{lemma:might-as-well-ni}]
    Let $\kappa = c\lambda_G$, $G' = (G\cup I)\setminus D$ and $H' = (H\cup I)\setminus D$ denote the two graphs obtained by applying the updates.
    In addition, let $C\subseteq V$ be a minimum cut of $G'$.
    We will show that $C$ is also a minimum cut of $H'$, which implies that the minimum cut of $G'$ is the same as the minimum cut of $H'$.

    Observe that $\mintcut_{G'}(C) \le k +\lambda_G <\tfrac{1+c}{2}\lambda_G$, since $C$ is a minimum cut of $G'$ and $k+l<\tfrac{c-1}{2}\lambda_G$.
    Furthermore, $\mintcut_{G}(C) \le  k + l +\lambda_G < \tfrac{1+c}{2}\lambda_G$, since $k+l<\tfrac{c-1}{2}\lambda_G$.
    Since $c>1$, we have that $\mintcut_G(C) < \kappa$, and thus the cut $C$ is preserved in any maximal $\kappa$-packing of forests of $G$, and thus also by $H$.
    Hence, all the edges of $E(C, V\setminus C)$ are preserved in $H$, and thus $\mintcut_{H'}(C) = \mintcut_{G'}(C)$.
\end{proof}
\begin{proof}[Proof of \Cref{lemma:might-as-well-nmc}]
    Let $G' = (G\cup I)\setminus D$ and $H' = (H\cup I)\setminus D$ denote the two graphs obtained by applying the updates.
    Let $C\subseteq V$ be a minimum cut of $G'$.
    We will show that with constant probability, no edge in $E_{G'}(C, V\setminus C)$ is contracted in $H'$.

    Observe that $\mintcut_{G'}(C) \le k +\lambda_G <\lambda_G+\tfrac{1}{20}\delta_G$, since it's a minimum cut of $G'$ and $k+l<\tfrac{1}{20}\delta_G$.
    Furthermore, $\mintcut_{G}(C) \le  k +\lambda_G + l < \lambda_G + \tfrac{1}{10}\delta_G$, since $k+l<\tfrac{1}{10}\delta_G$.
    By the definition of an NMC sparsifier, with constant probability no edge in $E_{G}(C, V\setminus C)$ is contracted in $H$, and thus also in $H'$.
    Finally, every edge inserted or deleted in $G'$ is also inserted or deleted in $H'$, and thus no edge in $E_{G'}(C, V\setminus C)$ is contracted in $H'$ with constant probability. 
\end{proof}

\end{document}